\documentclass{article}

\usepackage{arxiv}
\usepackage{amsthm}
\usepackage[utf8]{inputenc} 
\usepackage[T1]{fontenc}    
\usepackage{hyperref}       
\usepackage{url}            
\usepackage{amsfonts}       
\usepackage{nicefrac}       
\usepackage{microtype}      
\usepackage{graphicx}

\usepackage{amsmath}
\usepackage{amssymb}
\usepackage{bm}
\usepackage{dsfont}
\usepackage{thmtools} 
\usepackage{thm-restate}
\usepackage{booktabs}
\usepackage{adjustbox}
\usepackage[useregional]{datetime2}

\allowdisplaybreaks

\newcommand\numberthis{\addtocounter{equation}{1}\tag{\theequation}} 

\newcommand{\Nset}{\mathbb{N}}

\newcommand{\iguald}{\stackrel{d}{=}} 
\newcommand{\cd}{\stackrel{d}{\longrightarrow}} 
\newcommand{\cp}{\stackrel{p}{\longrightarrow}} 

\newcommand{\ind}{\mathds{1}} 
\newcommand{\ho}{\text{H$_0$}} 
\newcommand{\ha}{\text{H$_1$}} 

\theoremstyle{plain}
\newtheorem{theorem}{Theorem}[section]
\newtheorem{lemma}[theorem]{Lemma}

\newtheorem{proposition}[theorem]{Proposition}

\theoremstyle{definition}
\newtheorem{definition}[theorem]{Definition}
\newtheorem{example}[theorem]{Example}

\theoremstyle{remark}
\newtheorem{remark}{Remark}

\usepackage{epstopdf}
\usepackage[caption=false]{subfig}
\usepackage[nolists, tablesfirst]{endfloat}
\usepackage[doublespacing]{setspace}
\usepackage[numbers]{natbib}

\begin{document}

\title{Integer-valued autoregressive process with flexible marginal and innovation distributions}

\author{
 Matheus B. Guerrero \\
  King Abdullah University of Science and Technology, Thuwal, Saudi Arabia\\
  Universidade Federal de Minas Gerais, Belo Horizonte, Brazil \\
  \texttt{matheus.bartologuerrero@kaust.edu.sa} \\
   \And
 Wagner Barreto-Souza \\
  King Abdullah University of Science and Technology, Thuwal, Saudi Arabia\\
  Universidade Federal de Minas Gerais, Belo Horizonte, Brazil\\
  \And
 Hernando Ombao \\
  King Abdullah University of Science and Technology, Thuwal, Saudi Arabia\\
}

\maketitle

\begin{abstract}
\noindent INteger Auto-Regressive (INAR) processes are usually defined by specifying the innovations and the operator, which often leads to difficulties in deriving marginal properties of the process. In many practical situations, a major modeling limitation is that it is difficult to justify the choice of the operator. To overcome these drawbacks, we propose a new flexible approach to build an INAR model: we pre-specify the marginal and innovation distributions. Hence, the operator is a consequence of specifying the desired marginal and innovation distributions. Our new INAR model has both marginal and innovations geometric distributed, being a direct alternative to the classical Poisson INAR model. Our proposed process has interesting stochastic properties such as an MA($\infty$) representation, time-reversibility, and closed-forms for the transition probabilities $h$-steps ahead, allowing for coherent forecasting. We analyze time-series counts of skin lesions using our proposed approach, comparing it with existing INAR and INGARCH models. Our model gives more adherence to the data and better forecasting performance.


\end{abstract}

\section{Introduction}\label{chap:intro}

Count time series data arise naturally in many areas such as queuing systems, finance, insurance theory, medicine, epidemiology, among others. Thus, it is not surprising that there is a wide interest in models for count time series that can flexibly capture the features of the data. 

A usual way to deal with time series data is the AutoRegressive Moving Average (ARMA) models \citep{boxjenkins}. However, these real-valued models are not ideal for analyzing integer-valued time series data because of the nature of the data. The pioneering INteger-valued AutoRegressive (INAR) process \citep{McKenzie85, McKenzie88, Alzaid87} based on the binomial thinning operator \citep{vanHarn79} and in assuming Poisson innovations is an alternative for the usual ARMA process. Generally speaking, a thinning-based operator INAR process \citep{Scotto15} $\{X_t\}_{t\in \Nset}$ has the following stochastic structure:
\begin{eqnarray*}
X_t=\alpha\circ X_{t-1}+\epsilon_t, \quad t\in \Nset,
\end{eqnarray*}
where ``$\circ$'' is the thinning operator, introduced in \citep{vanHarn79}, replacing the usual multiplication; and $\{\epsilon_t\}_{t\in\mathbb N}$ is a sequence of independent and identically distributed (iid) integer-valued random variables called innovations. \citep{Alzaid87} gave an interpretation for the above INAR process in a population context, with $X_t$ denoting the count for the total population (of some region) at time $t$, $\alpha\circ X_{t-1}$ being the survivors from time $t-1$ and $\epsilon_t$ expressing the immigration.

Estimation of the parameters and forecasting for the Poisson first-order INAR (denoted by PINAR(1)) process (given by the above stochastic representation with Poisson innovations) were addressed in \citep{Freeland04a}, \citep{Freeland04b} and \citep{Freeland05}. High-order INAR processes were proposed in \citep{Alzaid90} and \citep{DuLi91}. 

Assuming the innovations to be Poisson distributed in a thinning-based operator INAR process implies that the marginals are also Poisson distributed. A limitation of this model is the fact that Poisson distribution can not accommodate underdispersion or overdispersion. This is a serious limitation because the model can produce results that are misleading. Alternatives to the PINAR(1) process have been proposed in the literature to overcome these situations. Some alternatives were introduced in \citep{McKenzie86} and \citep{Alosh92} with the processes having negative binomial and geometric marginals, respectively. A general class of mixed Poisson INAR(1) processes including the negative binomial and Poisson inverse-Gaussian models was recently introduced in \citep{Souza17}. Another important INAR process, called new geometric INAR (NGINAR) process, was proposed \citep{Ristic09}. This model is based on the negative binomial operator \citep{Aly94} and the marginals are geometric distributed. 

In some cases, there is inflation or deflation of zeros that can not be captured for example by the Poisson INAR process. In fact, the skin lesions dataset that is being analyzed contains many zeros. To work around such a limitation of the Poisson INAR process, \citep{Jazi12} and \citep{Souza15} proposed zero-inflated/deflated INAR models. INAR processes for dealing with negative integer-valued time series are due to \citep{freeland10}, \citep{Wang11} and \citep{SouzaBourguignon15}. Another extension allowing for random correlation parameter was introduced in \citep{Zheng06}, \citep{Zheng07}, \citep{Gomes09}, \citep{Wang11} and \citep{Zhao15}. Bayesian inference for INAR models is discussed in \citep{McCabe05} and \citep{Bisaglia16}, with special attention for prediction. More recently, \citep{Scotto2017} proposed a max-INAR model to deal with time series with sudden large counts, often caused by an extreme event, followed by monotone decreasing recovery phase. There is an extensive literature about INAR processes in the book \citep{weiss18} for a comprehensive treatment of this topic. 

A common feature of the INAR processes discussed above is that they are defined by fixing some operator (usually the classical binomial thinning) and also by specifying (a) the marginal or (b) the innovation distribution. Approach (b) is the one commonly done in practice (few exceptions may be found in \citep{weiss18}, page 24). With this approach, the limitation is that the marginal quantities cannot be obtained in simple form and that deriving a closed form for the transition probability is extremely difficult. Moreover, it is hard to justify the choice for the operator to be used.

The goal of this paper is to develop a novel statistical model that overcomes the limitations of the current models. More specifically, we propose a major refinement of INAR processes by specifying the marginal and innovation distributions rather than to fix some operator. With our new proposed approach, the operator to be considered as a consequence of having pre-specified the distribution of the marginals and innovations. For example, by specifying the marginals and innovations to be Poisson distributed we obtain the appropriate thinning operator that gave rise to desired marginal and innovations distributions. In particular, the advantages of our proposed approach are: (i) It provides an easy mechanism for simulating a stationary Markov chain without the need of numerical approximations. (ii)  Different operators rather than the classical thinning have different interpretations. If the operator admits values greater than 1, the operations in this case may be seen as a reproduction mechanism. (iii) It allows the user to specify the marginal distribution of a process which is important for incorporating a physical justification in the model (see \citep{livetal18}). 

For this, we define an INAR process with both marginals and innovations geometric distributed. A remarkable property of our geometric process is that it is time-reversible like the Poisson INAR and Gaussian AR(1) processes. We are unaware of other time-reversible INAR processes. Thus, one additional advantage of our proposed process process is that we derive closed-forms for the transition probabilities $h$-steps ahead, allowing for coherent forecasting without the need for numerical approximations.

We also compare our model to some of its natural competitors, such as the New Geometric INAR(1) model \citep{Ristic09} and the Integer-valued Generalized autoregressive conditional heteroscedastic (INGARCH) \citep{hei03,feretal06,foketal09,zhu11,silvaRB19}. Among the INAR competitors, we draw attention to the fact that the model in \citep{Ristic09} has geometric marginals, and it is defined through a negative binomial thinning operator. Yet, the resulting innovation no longer has a geometric distribution. On the INGARCH and INAR comparison, as discussed in \citep{weiss18}, there are pros and cons of each methodology. Beyond the INAR and INGARCH approaches, \citep{Cui09} and \citep{Lund15} proposed an alternative way of modeling count-time series by a linear combination of stationary renewal processes.

This paper is organized as follows. In Section \ref{chap:model}, a new perspective on INAR processes is presented. We construct a new INAR process with both marginals and innovations geometric distributed. In Section \ref{chap:properties} we derive main statistical properties. 
We develop estimators for our proposed geometric model and establish their asymptotic behavior in Section \ref{chap:estimation}. Furthermore, Monte Carlo simulations are provided to evaluate the finite-sample performance of the estimators and prediction. To study bovine skin lesions count data, which is a potential public health hazard, we use the proposed model. We report these interesting findings in Section \ref{chap:data}. The Appendix \ref{chap:demos} contains the proofs of the main propositions of the paper. All other proofs omitted in the text are given in the Supplementary Material (Appendix \ref{chap:sup}).

\section{New INAR processes}\label{chap:model}

An INteger-valued AutoRegressive (INAR) process $\{X_{t}\}_{t\in\mathbb N}$ can be defined by the following stochastic equation:
\begin{equation}\label{eq:process}
X_t = \bm{\theta}\star X_{t-1} + \varepsilon_{t},\quad t\in \mathbb N,
\end{equation}
where $\{ X_{t} \}_{t \in \Nset}$ is usually assumed to be stationary with marginal discrete distribution, $\star$ is a operator given by $\bm{\theta}\star X_{t-1} \equiv \sum_{i=1}^{X_{t-1}}{G_{i}}$, with $\{ G_{i} \}_{i\in\mathbb N}$ being a sequence of iid non-negative integer-valued random variables with common distribution $G$ depending on the parameter vector $\bm{\theta}$ and $\{ \varepsilon_{t} \}_{t\in\mathbb N}$ is a sequence of iid integer-valued random variables called innovations, which are independent of the sequence $\{G_i\}_{i\in\mathbb N}$. Furthermore, $X_{t-h}$ is independent of $\varepsilon_{t}$, $\forall h \geq 1$.

\begin{remark}
At each time $t$, a new thinning operation is performed independently of the past. Thus, the correct would be indexing the operator in time as $\bm{\theta}_t$. However, we avoid to do it to simplify the notation \citep{weiss18}.
\end{remark}

Under the stochastic structure in Eq. \eqref{eq:process}, for a chosen thinning operator, one of the two approaches may be adopted: (i) select the marginal distribution and derive the distribution of the innovation;  or (ii) choose the innovation distribution, then derive the distribution of the marginal. In the first approach (less common in the literature), despite given flexibility to a practitioner to choose the marginal distribution, in general, it is complicated to show that the process is well defined. While in the second approach, the practitioner cannot select a suitable distribution to the marginal, which often implies difficulties in obtaining the properties of the process. In fact, most of the time, the marginal distribution cannot be obtained explicitly \citep{Souza15}. Since the literature extensively considers these two approaches, \citep{weiss18} discusses all the technical nuances in choosing each one. A common drawback of both methods is that the operator must be chosen beforehand. Hence, in this paper, our main idea is to fix the marginal and innovation distributions, which provides flexibility to the practitioner and enables more control over the behavior of the count phenomena. As noted above, the operator is merely a tool for achieving the desired (user-specified) marginal and innovation distributions. 

Let $\Psi_Y(s)=E(s^Y)$ be the probability generating function (pgf) of a discrete random variable $Y$ for $s$ belonging some interval containing the value 1. By assuming $\{ X_{t} \}_{t \in \Nset}$ be a stationary process, from Eq. \eqref{eq:process} we obtain that
\begin{equation}
\Psi_{X}(s) = \Psi_{X}\left( \Psi_{G}(s) \right)\Psi_{\varepsilon}(s),
\label{eq:pgfgeneral1}
\end{equation} 
where $X_t\stackrel{d}{=}X$, $\varepsilon_t\stackrel{d}{=}\varepsilon$ and $G_i\stackrel{d}{=}G$. 
Additionally, if we assume $\Psi_{X}(s)$ is invertible, we have from Eq. \eqref{eq:pgfgeneral1} that
\begin{equation}
\Psi_{G}(s) = \Psi_{X}^{-1}\left( \frac{\Psi_{X}(s)}{\Psi_{\varepsilon}(s)} \right).
\label{eq:pgfinv}
\end{equation}

Therefore, once the distributions of $X_{t}$ and $\varepsilon_{t}$ are specified, Eq. \eqref{eq:pgfinv} represents a mechanism to obtain the distribution of $G$ through its pgf, as long as $\Psi_{G}(s)$ is a pgf from a proper discrete distribution. In what follows we discuss the Poisson case which is well-known in the literature \citep{McKenzie85, Alzaid87} and a new geometric case \citep{Ristic09}.

\begin{example} Let $X_{t}\sim \text{Poisson}(\mu)$, for $t\in \Nset$. Thus it is most natural to assume the innovations, $\varepsilon_t$, to be also Poisson distributed; $\varepsilon_{t}\sim \text{Poisson}((1-\alpha)\mu)$, for $0 < \alpha < 1$. 

We shall derive the necessary operator that gives the desired distributions. 

We have, for $s>0$, that $\Psi_{X}(s)=e^{-\mu(1-s)}$ and $\Psi_{\varepsilon}(s)=e^{-(1-\alpha)\mu(1-s)}$. Besides, $\Psi_{X}^{-1}(s) = 1 + \dfrac{1}{\mu}\log s$. From Eq. \eqref{eq:pgfgeneral1}, we obtain that
$$\Psi_G(s) = 1 - \alpha + \alpha s, \quad s>0,$$
which is the pgf of a Bernoulli distribution with success parameter $\alpha \in (0,1)$. With these conditions and under Eq. \eqref{eq:process}, the resulting operator is the well-known binomial thinning.
\end{example}

\begin{example}\label{ex_geom} Assume that $X_{t}\sim \text{Geo}(\mu)$, $\mu > 0$. More specifically, 
$$P(X_t=x) = \dfrac{\mu^{x}}{(1+\mu)^{x+1}},\, x = 0, 1, 2, \ldots.$$ 
Hence, 
$$E(X_{t}) = \mu\; \text{ and }\; Var(X_{t}) :=\sigma^{2} = \mu(1+\mu).$$ 

Also, assume the same distribution for the innovations, but with different mean; $\varepsilon_{t}\sim \text{Geo}((1-\alpha)\mu)$, for $t\in\Nset$ and $0 < \alpha < 1$. 

Denote $\mu_{\varepsilon}\equiv E(\varepsilon_{t})  = (1-\alpha)\mu$.

Which should be the operator that ensures the pre-specified distributions for the marginal and innovations?

For $|s| < \dfrac{1 + \mu}{\mu}$, we have that
\begin{equation}
\Psi_{X}(s) = \frac{1}{1+\mu(1-s)} \quad \Rightarrow \quad 
\Psi_{X}^{-1}(s) = 1 - \dfrac{1}{\mu}\left( \dfrac{1-s}{s} \right).   
\label{eq:pgfx}
\end{equation}

Analogously, 
\begin{equation}
\Psi_{\varepsilon}(s) = \frac{1}{1+\mu_{\varepsilon}(1-s)},\,\, |s| < \frac{1 + \mu_{\varepsilon}}{\mu_{\varepsilon}}.
\label{eq:pgfe}
\end{equation}

By substituting Eq. \eqref{eq:pgfx} and Eq. \eqref{eq:pgfe} in Eq. \eqref{eq:pgfinv} we have
\begin{equation}
\Psi_{G}(s) =\dfrac{1 + \left( 1 - \dfrac{\alpha}{(1-\alpha)\mu} \right)(1-\alpha)\mu(1-s)}{1+(1-\alpha)\mu(1-s)},\,\, |s| < \frac{1 + \mu_{\varepsilon}}{\mu_{\varepsilon}}.  \numberthis 
\label{eq:pgfzmg}	
\end{equation}

Therefore, Eq. \eqref{eq:pgfzmg} shows that $\Psi_{G}(s)$ is a pgf of a Zero-Modified Geometric (ZMG) distribution with parameters 
$$\pi = 1 - \frac{\alpha}{(1-\alpha)\mu} = 1 -\frac{\alpha}{\mu_{\varepsilon}} \quad \text{and} \quad \mu_{\varepsilon} = (1-\alpha)\mu,$$
both depending on the parameter vector $\bm{\theta} = (\mu, \alpha)$. 

We denote $G \sim \text{ZMG}(\pi, \mu_{\varepsilon})$ and its probability mass function (pmf) is given by
$$P(G = k)=
\left\{ 
\begin{array}{ll}
\pi + (1-\pi)\left(\dfrac{1}{1+\mu_{\varepsilon}}\right), &\text{for} \; k=0,\\[6pt]
(1-\pi)\dfrac{\mu_{\varepsilon}^{k}}{(1+\mu_{\varepsilon})^{k+1}}, &\text{for} \; k = 1, 2, \ldots . 
\end{array}
\right.$$

The mean and variance of $G$ are, respectively, 
$$E(G) = \mu_{\varepsilon}(1-\pi) = \alpha$$
and
$$Var(G) = \mu_{\varepsilon}(1-\pi)\left[ 1 + \mu_{\varepsilon}(1+\pi)\right] = (1 + 2\mu)(1-\alpha)\alpha.$$

Note that for $\pi \in ( -1/\mu_{\varepsilon}, 0 )$ and $\pi \in (0,1)$, we have a zero-deflated model and a zero-inflated model with respect to the geometric distribution, respectively. The ZMG distribution has the geometric distribution as a special case by taking $\pi = 0$.

We call the induced operator as {\it Zero-Modified Negative Binomial} (ZMNB) operator, denoted along the paper by $\star$, where a justification for this name comes from the conditional distribution given in the Proposition \ref{prop:sumzmg_cond}.
\end{example}

We are now able to define an INAR process with both marginals and innovations geometric distributed as follows.

\begin{definition}(Geo-INAR(1) process)
We say that a sequence $\{X_t\}_{t\in\mathbb N}$ is a first-order geometric INAR process if satisfies the stochastic equations given in Eq. \eqref{eq:process} with $\star$ being the Zero-Modified Negative Binomial operator given in the Example \ref{ex_geom}. 
\end{definition}

\begin{remark}
Our proposed geometric INAR(1) model will be denoted by Geo-INAR(1) process along this paper in order to be distinguished from the New geometric INAR(1) (NGINAR) process by \citep{Ristic09}. Also, note that our process is Markovian by definition.
\end{remark}

In Figure \ref{fig:paths}, we present simulated trajectories of our Geo-INAR(1) process when $\mu=1$, $5$ and $\alpha = 0,1$, $0.5$, $0.7$. Note that $\mu$ represents the mean of the count process, while $\alpha$ is related to the autocorrelation, as will be demonstrated in Proposition \ref{prop:transition}.

\begin{figure}[!htp]
	\centering
	\includegraphics[width=\textwidth]{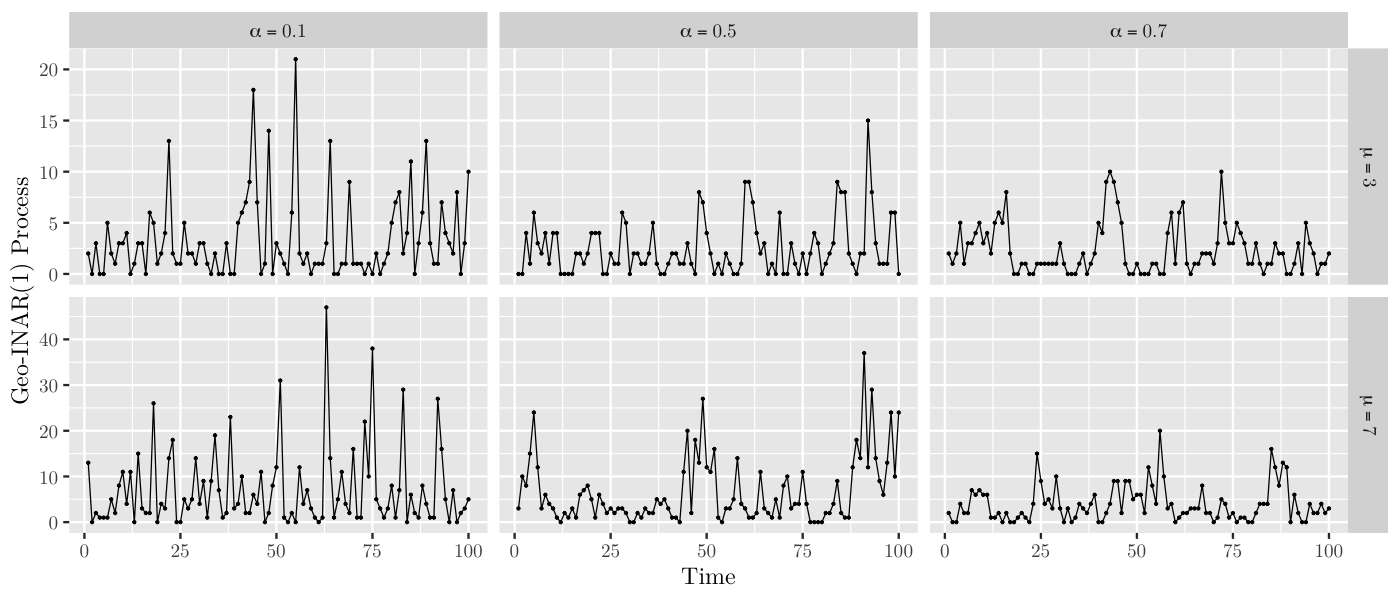}
	\caption{Simulated trajectories of the Geo-INAR(1) process for $\mu = 1, 5$ and $\alpha=0.1, 0.5, 0.7$. Sample size is equal to 100.}
	\label{fig:paths}
\end{figure}

To obtain some properties and quantities for the Geo-INAR process, we need some properties related to the $\star$ operator. In the following propositions, we state some important properties of this operator. For instance, the moments derived in Proposition \ref{prop:basic} are required in many of the proofs in the text. The proofs for Proposition \ref{prop:sumzmg_cond} and Proposition \ref{prop:transition} can be found in the Appendix \ref{chap:demos}. All the other proofs omitted in the text are given in the Supplementary Material (Appendix \ref{chap:sup}).

\begin{restatable}{proposition}{propgenerallabel}
\label{prop:basic}
Let $\{ G_{i} \}_{i=1}^\infty \stackrel{iid}{=} G \sim \text{ZMG}(\pi, \mu_{\varepsilon})$. Let $X$ and $Y$ be non-negative integer-valued random variables not necessarily independent of each other, but independent of the sequence $\{ G_{i} \}_{i=1}^\infty$. We have that
\begin{enumerate}
	\item [i.] $E(\bm{\theta} \star X) =  E(G)E(X);$
	\item [ii.] $E((\bm{\theta} \star X)^{2}) = Var(G)E(X) + E^{2}(G)E(X^{2});$
	\item [iii.] $E\left((\bm{\theta} \star X)Y\right) = E(G)E(XY);$
	\item [iv.] $Var(\bm{\theta} \star X) = Var(G)E(X) + E^2(G)Var(X);$
	\item [v.] $Cov(\bm{\theta} \star X, X) = E(G)Var(X).$
\end{enumerate}
\end{restatable}

\begin{restatable}{proposition}{propsumzmgcondlabel}
\label{prop:sumzmg_cond} 
For $\mu > 0$ and $0< \alpha < \mu / (1 + \mu)$, the conditional distribution of $\bm{\theta}\star X|X = x$ is a Zero-modified negative binomial distribution with parameters $\pi$, $\mu_{\varepsilon}$ and $x \geq 1$, and with pmf given by
\begin{equation*}
P \left( \bm{\theta}\star X = k | X = x  \right) = \begin{cases}
\pi_{\star}^{x}, k=0,\\
\displaystyle{
\sum\limits_{i=1}^{k}{\mbox{A}_{i}^{x}(\pi_{\star})\mbox{B}_{i}^{k}(p) }}, k \geq 1,
\end{cases}
\end{equation*}
where
$$\mbox{A}_{i}^{x}(\pi_{\star}) := \displaystyle{\binom{x+i-1}{i}\pi_{\star}^{x}(1-\pi_\star)^{i}}\;
\text{ and }\; 
\mbox{B}_{i}^{k}(p) := \displaystyle{\binom{k-1}{i-1}p^{k-i}(1-p)^{i}}.$$
Also,
$$\pi_{\star} = \pi + (1-\pi)\dfrac{1}{1 + \mu_{\varepsilon}} = 1 - \dfrac{\alpha}{1+\mu_{\varepsilon}} \quad \text{and} \quad p = \dfrac{\mu_{\varepsilon} - \alpha}{1 + \mu_{\varepsilon} - \alpha}.$$

Furthermore, 
$$E(\bm{\theta}\star X = k | X = x) = x E(G)\; \text{ and }\; Var(\bm{\theta}\star X = k | X = x) = x Var(G).$$

We denote $ \bm{\theta}\star X| X = x \sim \text{ZMNB}(\pi, \mu_{\varepsilon}, x)$.
\end{restatable}

Note that the ZMNB distribution has the negative binomial distribution as a special case when $\pi = 0$.  

\begin{remark}
By convention, if $X = 0$, then $\bm{\theta}\star X = 0$, implying that 
$$P \left( \bm{\theta}\star X = k | X = 0  \right) =
\left\{
\begin{array}{l} 
    1,\; k=0,\\
    0,\; k \geq 1.
\end{array}
\right.
$$
Thus, the conditional pmf in Proposition \ref{prop:sumzmg_cond} is well defined for all values of $X$. 
\end{remark}

\begin{remark}
\citep{kolev2000} also obtain an expression for the pmf of a ZMNB distribution; however, it is computationally intractable. The expression we get is much simpler and easier to implement computationally.
\end{remark}

For our Geo-INAR(1) process, we can obtain the marginal distribution of $\bm{\theta}\star X$. Since $\Psi_{\bm{\theta}\star X}(s) = \Psi_{X}\left( \Psi_{G}(s) \right)$ and by using Eq. \eqref{eq:pgfx} and Eq. \eqref{eq:pgfzmg}, we obtain that
\begin{equation*}
\Psi_{\bm{\theta}\star X}(s) = \dfrac{1+(1-\alpha)\mu(1-s)}{1+\mu(1-s)}.
\end{equation*}

That is, $\bm{\theta}\star X \sim \text{ZMG}(1-\alpha, \mu)$. 

In particular, we get $E(\bm{\theta}\star X_{l}) = \mu - \mu_{\varepsilon}$ and $Var(\bm{\theta}\star X_{l}) = \sigma^{2} - \sigma_{\varepsilon}^{2}$.

Another interesting feature of our process is that the ZMG distribution is closed by associating operations. Consider successive operations of the $\star$ operator with the random variable $X\sim \text{Geo}(\mu)$, $\mu > 0$, such as 
$$\bm{\theta}_{h}\star \bm{\theta}_{h-1}\star \cdots \star \bm{\theta}_{1}\star X.$$

The objective is to determine the distribution of this random variable. In this context, we allowed the counting series related to the operation ``$\bm{\theta}_i \star$'' varies with the index $i$, i.e, $\forall i = 1, 2, \ldots h$,
$$\bm{\theta}_i \star X = \sum_{j=1}^{X}{G_{i,j}}, \;\; \text{where}\;\; G_{i, j} \stackrel{iid}{\sim} \text{ZMG}(\pi_i, {\mu_{\varepsilon}}_i),\;\; \forall i,j.$$

The case where all the $h$ random variables $G_i$ are iid ZMG distributed with parameters $\pi = 1 - \alpha/{\mu_{\varepsilon}}$ and ${\mu_{\varepsilon}} = (1-\alpha)\mu$, is written as
$${\bm{\theta}}^{(h)}\star X = \underbrace{\bm{\theta} \star \bm{\theta} \star \ldots \star  \bm{\theta}}_{h \text{ times}} \star X.$$
Note that, in this case, all operators $\star$ have the same counting series $G$, where $G \sim \text{ZMG}(\pi, {\mu_{\varepsilon}})$. Therefore, due to the Lemma \ref{prop.sum} given in the Supplementary Material (Appendix \ref{chap:sup}), we have that
$${\bm{\theta}}^{(h)}\star X = \sum_{j=1}^{X}G_{j}^{(h)},\; \text{where}\; G_{j}^{(h)}\sim \text{ZMG}\left(1-\frac{\alpha^h}{\left(1-\alpha^h\right)\mu}, \left(1-\alpha^h\right)\mu \right), \forall j.$$

We summarize these properties in the following proposition.

\begin{restatable}{proposition}{propassociatedoperatorslabel}
\label{theo:prodZMG}
Let $X\sim \text{Geo}(\mu)$, $\mu > 0$. For $\bm{\theta}_i = (\mu, \alpha_i)$, $0 < \alpha_i < 1$ let
$$G_{i} \sim \text{ZMG}(\pi_i, {\mu_{\varepsilon}}_i),\;\; \pi_i = 1- \alpha_{i}/{\mu_{\varepsilon}}_i\;\; \text{and}\;\; {\mu_{\varepsilon}}_i = (1-\alpha_i)\mu,\; \forall i = 1, 2, \ldots h.$$

Then,
$$\bm{\theta}_{h}\star \bm{\theta}_{h-1}\star \cdots \star \bm{\theta}_{1}\star X \sim \text{ZMG}\left(1 - \prod_{i=1}^{h}\alpha_i, \mu \right).$$

The case where all the $h$ random variables $G_i$ have the same parameters $\pi = 1- \alpha/{\mu_{\varepsilon}}$ and ${\mu_{\varepsilon}} = (1-\alpha)\mu$ leads to
$${\bm{\theta}}^{(h)}\star X \sim \text{ZMG}\left(1-\alpha^{h}, \mu\right).$$
\end{restatable}

The properties presented in this section play an essential role in the next sections, where we obtain an MA representation for our process, its conditional moments, as well as the transition probabilities $h$-steps ahead, an important mechanism for coherent forecasting.

\section{Stochastic properties of the Geo-INAR(1) process}\label{chap:properties}

We here present some stochastic properties and also some important quantities like joint probability function and conditional moments for our proposed geometric INAR(1) process.
We begin by presenting some stochastic representations.

\begin{restatable}{proposition}{theoigualIDlabel}
	\label{theo:igualD}
	Let $\{X_t\}_{t\in\mathbb N}$ be a Geo-INAR(1) process. Then, we have that
	\begin{equation*}
	X_{t+h} \iguald {\bm{\theta}}^{(h)}\star X_{t} + \varepsilon_{t}^{(h)},    
	\end{equation*}
	where $\varepsilon_{t}^{(h)}\equiv\sum\limits_{j=0}^{h-1}{\bm{\theta}}^{(j)}\star\varepsilon_{t +h - j}$, for $t, h \geq 1$.
\end{restatable}

\begin{proof}
	By applying Proposition \ref{theo:prodZMG}, we obtain that ${\bm{\theta}}^{(h)}\star X_{t} \sim \text{ZMG}(1-\alpha^{h}, \mu)$ and $\varepsilon_{t}^{(h)} \sim \text{Geo}((1-\alpha^{h}) \mu)$. Now, the desired result follows by using Lemma \ref{prop.sum} from the Supplementary Material (Appendix \ref{chap:sup}).
\end{proof}

\begin{restatable}{proposition}{DGINARasMAinftylabel}
	\label{prop:DGINARasMAinfty}
	A Geo-INAR(1) process $\{X_t\}_{t\in\mathbb N}$ can be expressed as a kind of moving average process of infinity order (MA($\infty$)) as follows:
	\begin{equation*}
	X_t \iguald  \sum_{j=0}^\infty \bm{\theta}^{(j)} \star \varepsilon _{t-j}.
	\end{equation*}
\end{restatable}

\begin{proof}
		By using Proposition \ref{theo:igualD}, we can write
		\begin{equation*}
		X_{t} \iguald {\bm{\theta}}^{(h+1)}\star X_{t-h-1} + \sum\limits_{j=0}^{h}{\bm{\theta}}^{(j)}\star\varepsilon_{t - j}\;, t \geq 1, \forall h\geq1.  
		\end{equation*}
        We have that ${\bm{\theta}}^{(h+1)}\star X_{t-h-1} \sim \text{ZMG}(1-\alpha^{h+1}, \mu)$, with expected value given by $E\left({\bm{\theta}}^{(h+1)}\star X_{t-h-1}\right) = \mu\alpha^{h+1}$. Then, $\forall \xi > 0$ and as $h \rightarrow \infty$,
		\begin{align*}
		\Pr \left( \left|{\bm{\theta}}^{(h+1)}\star X_{t-h-1} - 0\right| > \xi \right) &= 
		\Pr \left({\bm{\theta}}^{(h+1)}\star X_{t-h-1} > \xi \right)\\
		&{\stackrel{\text{Markov}}{\leq}} \; 
		\frac{E\left({\bm{\theta}}^{(h+1)}\star X_{t-h-1}\right)}{\xi} = \frac{\mu\alpha^{h+1}}{\xi} \rightarrow 0.
		\end{align*}
		
		Therefore, ${\bm{\theta}}^{(h+1)}\star X_{t-h-1} \cp 0$ as $h\rightarrow\infty$. On the other hand, 
		$$\varepsilon_{t}^{(h+1)} = \sum\limits_{j=0}^{h}{\bm{\theta}}^{(j)}\star\varepsilon_{t - j}  \cd \sum\limits_{j=0}^{\infty}{\bm{\theta}}^{(j)}\star\varepsilon_{t - j} \stackrel{d}{=} X_{t}, \quad \mbox{as}\,\, h\rightarrow\infty.$$ 
		
		To see this, note that $\varepsilon_{t}^{(h+1)} \sim \text{Geo}\left((1-\alpha^{h+1})\mu\right)$. Thus, for all $y\in\mathbb N$,
		
		\begin{align*}
		\Pr\left(\varepsilon_{t}^{(h+1)}\leq y\right) &= 1 - \left( \frac{(1-\alpha^{h+1})\mu}{1+(1-\alpha^{h+1})\mu} \right)^{y+1} \Rightarrow \\
		\lim_{h \rightarrow \infty} \Pr\left(\varepsilon_{t}^{(h+1)}\leq y\right) &= 1 - \left(\frac{\mu}{1+\mu}\right)^{y+1} = P(X_t\leq y).
		\end{align*}

        With the above results, we derive the desired MA($\infty$) representation of the Geo-INAR(1) process.

\end{proof}

\begin{remark}
Proposition \ref{prop:DGINARasMAinfty} guarantees that the proposed Geo-INAR(1) process can be represented as a kind of MA$(\infty)$ process like the Gaussian AR(1) and PINAR(1) processes.
\end{remark}

In the following proposition, we present the transition probabilities $1$-step ahead, $p_{ij} \equiv P(X_t = j | X_{t-1} = i)$, $i$, $j \in \Nset$, and also expressions for the conditional mean and variance, besides the autocorrelation function of the Geo-INAR(1) model. 

\begin{restatable}{proposition}{proptransitionlabel}
	\label{prop:transition} 
	Let $\{X_t\}_{t\in\mathbb N}$ be a Geo-INAR(1) process. Then, for $\mu>0$ and $\alpha \in (0, \mu/(1+\mu))$:
	\begin{enumerate}
		\item [i.] $\forall j\in\Nset$,
		\begin{equation*}
		p_{ij} = \begin{cases}
		p_{\varepsilon}(j), i = 0\\
		p_{\varepsilon}(j)\displaystyle{\left[\pi_{\star}^{i} + {\sum_{m=1}^{j}}{\sum_{l=1}^{m}}{\left(1 + \frac{1}{\mu_\varepsilon}\right)^{m}\mbox{A}_{l}^{i}(\pi_{\star})\mbox{B}_{l}^{m}(p)}\right]}, i\geq1.
	    \end{cases}
	    \end{equation*} 
	Here, $p_{\varepsilon}(\cdot)$ is the pmf of the innovations with
	$$\pi_{\star} = 1 - \dfrac{\alpha}{1+\mu_{\varepsilon}}\; \text{ and }\; p = \dfrac{\mu_{\varepsilon} - \alpha}{1 + \mu_{\varepsilon} - \alpha}.$$
	
	$\mbox{A}_{l}^{i}(\pi_{\star})$ and $\mbox{B}_{l}^{m}(p)$ are defined in Proposition \ref{prop:sumzmg_cond}.
	\item [ii.] $E(X_{t+1}|X_t) =  \alpha X_t + (1-\alpha)\mu$.
	\item [iii.] $Var(X_{t+1}|X_t) =  \left[ (1+2\mu)(1-\alpha)\alpha \right]X_t + \sigma_{\varepsilon}^{2}$.
	\item [iv.]The autocorrelation function is $\rho_h \equiv corr(X_{t+h},X_t)=\alpha^{h}$, for $h\in\Nset$.
	 	\end{enumerate}
\end{restatable}

It is possible to generalize this results. For the conditional mean and conditional variance, we use induction methods to obtain the recurrence relations:
\begin{itemize}
    \item[] $E(X_{t+h}|X_t) = E(G)E(X_{t+(h-1)}|X_t) + \mu_\varepsilon,$
    \item[] $Var(X_{t+h}|X_t) = Var(G)E(X_{t+(h-1)}|X_t) + E^{2}(G)Var(X_{t+(h-1)}|X_t) + \sigma_{\varepsilon}^{2}$
\end{itemize}

Then, solving the difference equations we can generalize the $h$-steps ahead conditional moments as follow:  
\begin{itemize}
	\item [] $\begin{aligned}[t] \phantom{=} E(X_{t+h}|X_t) &= \alpha^{h}X_t + \dfrac{1 - \alpha^{h}}{1 - \alpha}\mu_\varepsilon,
	\end{aligned}$
	\item [] $\begin{aligned}[t] &\phantom{=} Var(X_{t+h}|X_t) = (1+2\mu)\left[\alpha^{h+1}(1-\alpha^{h-1}) + \alpha^{2h-1}(1-\alpha)	\right]X_t + \\[7pt]
	&\phantom{=}\quad	(1+2\mu)\left[\alpha^{h+1}\dfrac{1-\alpha^{h-1}}{1-\alpha} + \alpha\dfrac{1-\alpha^{2h-2}}{1-\alpha^{2}}\right]\mu_\varepsilon + \left[\alpha^{2h-2} + \dfrac{1-\alpha^{2h-2}}{1-\alpha^{2}}\right]\sigma_{\varepsilon}^{2}.
	\end{aligned}$
\end{itemize}

Note that in the expressions above, when $h \rightarrow \infty$, we have that
\begin{enumerate}
	\item [] $\begin{aligned}[t] \phantom{=} E(X_{t+h}|X_t) \; &\rightarrow \; 0 X_t + \dfrac{1}{1-\alpha}\mu_\varepsilon = \mu = E(X_t),
	\end{aligned}$
	
	\item [] $\begin{aligned}[t] \phantom{=} Var(X_{t+h}|X_t) \; &\rightarrow \; (1+2\mu)0X_t + (1+2\mu)\dfrac{\alpha}{1-\alpha^2}\mu_\varepsilon + \dfrac{1}{1-\alpha^2}\sigma_{\varepsilon}^{2} \\[5pt]
	&= \dfrac{1}{1-\alpha^2}(1-\alpha)(1+\alpha)\mu(1+\mu) = \mu(1+\mu) = Var(X_t)
	\end{aligned}$
\end{enumerate}
which are, as expected, the unconditional mean and unconditional variance, respectively.

The transition probabilities $h$-steps ahead are essential to obtain coherent forecast.  In the context of integer-valued time series, if we use the conditional mean to obtain predictions, the predicted values may not belong to the parametric space. In this sense, these predictions are not coherent. To get predicted values $h$-steps ahead, we need to consider the conditional pmf $h$-steps ahead $P(X_{t+h} = j|X_{t} = i)$, $j=0$, $1$, $2$, $\ldots$, for a fixed $i$, and then use its mode or median as point forecast. However, most of the time it is hard to obtain closed forms of this pmf, as stated by \citep{weiss18}. For our Geo-INAR(1) process we have the closed forms given by Proposition \ref{prop:transHahead}, a generalization of Proposition \ref{prop:transition} - item (i).

\begin{restatable}{proposition}{transHaheadlabel}
	\label{prop:transHahead} 
	The transition probabilities of the Geo-INAR(1) process $h$-steps ahead, $\forall j \in \Nset$, are:
	\begin{equation*}
	p_{ij}^{(h)} = \begin{cases}
	p_{\varepsilon^{(h)}}(j), i = 0\\
	p_{\varepsilon^{(h)}}(j)\displaystyle{ \left[
	\left(\pi_{\star}^{(h)}\right)^{i} + \sum_{m=1}^{j}\sum_{l=1}^{m}{\left(1 + \frac{1}{\mu_\varepsilon^{(h)}}\right)^{m}\mbox{A}_{l}^{i}\left(\pi_{\star}^{(h)}\right)\mbox{B}_{l}^{m}\left(p^{(h)}\right)}\right]}, i \geq 1,
	\end{cases}
	\end{equation*}
    \noindent where $p_{\varepsilon^{(h)}}(\cdot)$ is the pmf of the random variable $\varepsilon^{(h)}$ defined in Theorem \ref{theo:igualD}. $\mbox{A}_{l}^{i}(\cdot)$ and $\mbox{B}_{l}^{m}(\cdot)$ are defined in Proposition \ref{prop:transition}. Also,
	$$\pi_{\star}^{(h)} = 1 - \frac{\alpha^h}{1+\mu_{\varepsilon}^{(h)}},\;\; p^{(h)} = \frac{\mu_{\varepsilon}^{(h)} - \alpha^h}{1 + \mu_{\varepsilon}^{(h)} - \alpha^h},\;\; \text{and}\;\; \mu_{\varepsilon}^{(h)}=(1-\alpha^h)\mu.$$
\end{restatable}

\begin{proof}
Based on the fact that
$$\left(X_t, X_{t-h} \right)  \iguald \left(\bm{\theta}^{(h)}\star X_{t-h} + \sum_{k=0}^{h-1}{\bm{\theta}^{(k)}\star \varepsilon_{t-j}}, X_{t-h} \right).$$
The result follows by using the same steps of the case where $h=1$. It is exactly the same calculations but with $\alpha^h$ replacing $\alpha$.
\end{proof}

\begin{restatable}{proposition}{propconditionalpgflabel}
	\label{prop:conditionalpgf}
	The conditional pgf of the Geo-INAR(1) process is
	\begin{align*}
	E\left( s^{X_{t+h}}\big| X_t \right) =
	&\Psi_{X}(s)\left[ \Psi_{X}\left( \frac{1 + \left[(1-\alpha^{h-1})\mu - \alpha^{h}\right](1-s)}{1 + (1-\alpha^{h-1})\mu(1-s)} \right) \right]^{-1}\\
	&\times\left( \frac{1 + \left[(1-\alpha^{h-1})\mu - \alpha^{h}\right](1-s)}{1 + (1-\alpha^{h-1})\mu(1-s)} \right)^{X_t}.
	\end{align*}
\end{restatable}

We now are going to discuss about the joint distribution of $(X_{t},X_{t-1})$. For a discrete bivariate random vector $(Z_{1}, Z_{2})$, its probability generating function is given by 
$$\Psi_{Z_{1}, Z_{2}}(s_{1}, s_{2}) = E\left(s_{1}^{Z_{1}} s_{2}^{Z_{2}} \right),$$
$s_1$ and $s_2$ belonging to some interval containing the value 1.
The joint pgf of $(X_{t}, X_{t-1})$ can be expressed by
$$\Psi_{X_{t}, X_{t-1}}(s_{1}, s_{2}) = \Psi_{\varepsilon}(s_{1})\Psi_{X}\left(s_{2}\Psi_{G}(s_{1})\right).$$

After some algebra, we obtain
\begin{align}
\Psi_{X_{t}, X_{t-1}}(s_{1}, s_{2}) &= \dfrac{\Psi_{\varepsilon}(s_{1})}{1 + \mu(1 - s_{2}\Psi_{G}(s_{1}))}\nonumber\\
                                    &= \dfrac{1}{1 + \mu\left[ (1-s_{1}) + (1 - s_{2}) + (\mu_{\varepsilon} - \alpha)(1-s_{1})(1-s_{2}) \right]},
\label{eq:jointpgf}
\end{align}
which is a pgf of a bivariate geometric distribution with parameters $c_1 = c_2 = \mu$ and $\gamma^{2} = \frac{\mu}{1+\mu}\alpha$ as given by \citep{jayakumar}. We denote $(X_{t}, X_{t-1}) \sim \text{BGD}(c_1, c_2, \gamma^{2})$. 

With the above result, we get an important stochastic property of our geometric INAR(1) process as follows.

\begin{proposition}
	The Geo-INAR(1) process is time-reversible.
\end{proposition}

\begin{proof}
The time-reversibility of the Geo-INAR(1) process follows from the symmetry in $s_{1}$ and $s_{2}$ of the joint pgf in Eq. \eqref{eq:jointpgf}.
\end{proof}

\begin{remark}
By the time-reversibility of the Geo-INAR(1) process, the future and the past can be ``swapped'', that is, from an inferential point of view, one can analyze data from past to future or from future to past to obtain the same results. The time-reversibility property improves the estimation procedures, as can be seen in \citep{Annis10}. It is worth mentioning that this property is extremely rare on INAR processes. To the best of our  knowledge, only the existing Poisson INAR(1) process is time-reversible among the INAR ones. We call attention that in the Poisson case, both marginals and innovations are in the same family of distributions (Poisson) as in our geometric model. We conjecture that time-reversibility is an innate feature of INAR(1) processes constructed under our new proposed perspective. We hope to investigate this conjecture in future research.
\end{remark} 

\section{Estimation and Monte Carlo simulation}\label{chap:estimation}

In this section, we provide three methods for estimating the parameters of our proposed Geo-INAR(1) process. Let $\bm{\theta} = (\mu, \alpha)^{T}$ be the parameter vector and assume that $\{ X_t = x_t \}_{t=1}^{n}$ is a realization of the Geo-INAR(1) process, where $n$ stands for the sample size. Also, we present Monte Carlo simulation studies to compare the performance of the proposed estimation methods.

\subsection{Conditional Least Squares estimators}

We begin by discussing estimation by Conditional Least Squares (CLS) method, which consists in minimizing the function 
$$Q_{n}(\bm{\theta}) = \sum_{t=2}^{n}{\left( X_t - E(X_{t}|X_{t-1}) \right)^{2}} = \sum_{t=2}^{n}{\left( X_t - \alpha X_{t-1} - (1-\alpha)\mu  \right)^{2}},$$
with respect to $\mu$ and $\alpha$.

By solving the system of equations $\partial Q_{n}(\bm{\theta})/\partial \bm{\theta} = 0$, we obtain the following estimators

\begin{align*}
\widehat{\alpha}_{cls} &= \dfrac{\displaystyle\sum_{t=2}^{n}{X_t X_{t-1}} - \dfrac{1}{n-1}\displaystyle\sum_{t=2}^{n}{X_t}\displaystyle\sum_{t=2}^{n}{X_{t-1}}}{\displaystyle\sum_{t=2}^{n}{X_{t-1}^{2}}-\dfrac{1}{n-1}\left(\displaystyle\sum_{t=2}^{n}{X_{t-1}}\right)^2}, \text{ and}\\
\widehat{\mu}_{cls} &= \dfrac{\displaystyle\sum_{t=2}^{n}{X_t}-\widehat{\alpha}_{cls}\displaystyle\sum_{t=2}^{n}{X_{t-1}}}{(n-1)(1-\widehat{\alpha}_{cls})}. 
\end{align*}

Proposition \ref{cls_prop} exhibits the asymptotic distribution of the CLS estimators.
\begin{restatable}{prop}{clsproplabel}
	\label{cls_prop} 
	Let $\widehat{\bm{\theta}}_{cls} = (\widehat{\mu}_{cls}, \widehat{\alpha}_{cls})^{\top}$ be the CLS estimators of the Geo-INAR(1) process. Then, $\widehat{\bm{\theta}}_{cls}$ is strongly consistent for $\bm{\theta}$ and its asymptotic distribution is given by 
	$$
	\sqrt{n}\left[ 
	\begin{array}{c}
	\widehat{\mu}_{cls} - \mu \\
	\widehat{\alpha}_{cls} - \alpha
	\end{array}
	\right] \cd 
	\text{N}\left(\left[
	\begin{array}{c}
	0 \\
	0 
	\end{array}
	\right],
	\left[
	\begin{array}{cc}
	\dfrac{\mu(1+\mu)(1+\alpha)}{1-\alpha} & (1+2\mu)\alpha \\
	(1+2\mu)\alpha                         & \dfrac{(1+\mu+2\mu)\sigma_{G}^{2} + \sigma_{\varepsilon}^{2}}{\mu(1+\mu)} 
	\end{array}
	\right]
	\right),
	$$
	where $\sigma_{G}^{2} = Var(G) = (1+2\mu)(1-\alpha)\alpha$.
\end{restatable}

\subsection{Yule-Walker estimators}

Yule-Walker (YW) estimators, which are based on the method of moments, are obtained by solving the Yule-Walker equations, which estimate the parameters through the sample autocorrelation and moments. We obtain an analytic expression for $\widehat{\alpha}_{yw}$ as function of the first-autocorrelation, thus, since $\alpha = corr(X_t,X_{t-1})$ and by using the fact $\mu = E(X)$, we obtain
\begin{align*}
\widehat{\alpha}_{yw} &= \frac{\displaystyle\sum_{t=2}^{n}{(X_t - \bar{X})(X_{t-1}-\bar{X})}}{\displaystyle\sum_{t=1}^{n}{(X_t - \bar{X})^2}}, \text{ and }\\ 
\widehat{\mu}_{yw} &= \bar{X} = \frac{1}{n}\displaystyle\sum_{t=1}^{n}{X_t}.
\end{align*}

Under our Geo-INAR(1) process, the Yule-Walker and CLS estimators are asymptotically equivalent. Since in our Geo-INAR(1) process the Yule-Walker and CLS estimators are exactly the same form of the Poisson INAR(1) process considering the alternative parameterization proposed by \citep{Joe96}, we can use directly Theorem 3 from \citep{Freeland05} to ensure this result.

\subsection{Maximum likelihood estimators}

The Maximum Likelihood (ML) estimators are obtained by maximizing the log-likelihood function. We have that the likelihood function can be written as
$$\mathcal{L}(\bm{\theta}) = P(X_1 = x_1, \ldots, X_n = x_n) = P(X_1 = x_1)\prod_{t=2}^{n}{P(X_t=x_t|X_{t-1}=x_{t-1})},$$
due to the Markovian property of our process. Hence, the log-likelihood function $\ell(\bm{\theta})=\log \mathcal{L}(\bm{\theta})$ is given by
$$\ell(\bm{\theta}) = x_{1}\log\mu - (1-x_{1})\log(1+\mu) + \sum_{t=2}^{n}{ \log p_{x_{t-1}x_{t}} },$$
where $p_{x_{t-1}x_{t}}$ is the transition probabilities provided in Proposition \ref{prop:transition}.

Then, to obtain the maximum likelihood estimators we need to solve the non-linear system of equations $\partial \ell(\bm{\theta})/\partial \bm{\theta} = 0$ through numerical methods implemented in many statistical software.

\subsection{Monte Carlo simulation}

In this section, we present a Monte Carlo simulation study to compare the three estimation methods proposed in the previous subsections for the Geo-INAR process by setting some values of $\mu$ and $\alpha$ and some sample sizes. We compute the root of the mean square error (RMSE) of the estimates given by
$$\text{RMSE}(\bm{\theta}) = \sqrt{\frac{1}{n}\sum_{t=1}^{n}{( \hat{\bm{\theta}}_t - \bm{\theta}_0 )^2}},$$
where $\bm{\theta}_0$ is the true value of the parameter under study and $\widehat{\bm{\theta}}$ is its estimate by a selected estimation method. Note that RMSE is an absolute measure and since we have several scenarios and we wish to make comparisons among them, it is needed to transform the RMSE into a relative measure. For this, we will consider the relative RMSE.

We set up four scenarios to the simulation. In the scenarios (a) to (d), we fixed the value of the mean parameter as $\mu = 5$ and the autocorrelation parameter as $\alpha = 0.1$, $0.3$, $0.5$, $0.7$. For each scenario we simulated different sample sizes ($n = 100$, $300$, $500$, $700$, $1000$) of the Geo-INAR(1) process, thereafter we applied the estimation methods to compare their results. This process was replicated 5000 times, following the calculation of sample mean of the estimates and the relative RMSE.

The results of the simulation study regarding configurations (a)-(d) are presented in Table \ref{tab:5simulacao}. First, we observe consistency of all estimators. On one hand, we can notice that the three estimation methods behave equivalently for $\mu$, the estimates obtained by each method are very close to each other. On the other hand, for $\alpha$, we notice that there is, virtually, no difference between the CLS and YW estimates since they are asymptotically equivalent. While the ML estimates are better than both CLS and YW in all scenarios. Further, for large values of $\alpha$, we note a better performance of the ML estimator over the CLS and YW estimators. In each scenario, the percentage mean improvement of the ML method in terms of relative RMSE for $\alpha$, in relation to the second best estimation method, is: (a) $2.1$\%, (b) $16.5$\%, (c) $26.8$\%, and (d) $34.9$\%.

\begin{table}[!hb]
	\centering
	\caption{Numerical results of the scenarios (a), (b), (c) and (d). The relative RMSE is displayed in parentheses.}
	\begin{adjustbox}{max width=\textwidth}
		\begin{tabular}{rclcl|clcl|clcl}
			\toprule
			$n$   & \multicolumn{2}{c}{$\hat{\mu}_{cls}$} & \multicolumn{2}{c|}{$\hat{\alpha}_{cls}$} & \multicolumn{2}{c}{$\hat{\mu}_{yw}$} & \multicolumn{2}{c|}{$\hat{\alpha}_{yw}$} & \multicolumn{2}{c}{$\hat{\mu}_{ml}$} & \multicolumn{2}{c}{$\hat{\alpha}_{ml}$} \\
			\midrule
			\multicolumn{12}{c}{a) True values: $\mu = 5$ and $\alpha = 0.1$}                             &  \\
			100   &      5.0433  & (0.1242) &      0.1236  & (0.8727) &      5.0426  & (0.1232) &      0.1224  & (0.8618) &      5.0427  & (0.1231) &      0.1325  & (0.9058) \\
			300   &      5.0077  & (0.0711) &      0.1024  & (0.5809) &      5.0081  & (0.0711) &      0.1020  & (0.5789) &      5.0082  & (0.0712) &      0.1049  & (0.5620) \\
			500   &      5.0013  & (0.0536) &      0.0988  & (0.4791) &      5.0012  & (0.0536) &      0.0986  & (0.4783) &      5.0012  & (0.0536) &      0.1002  & (0.4603) \\
			700   &      5.0043  & (0.0465) &      0.0987  & (0.4168) &      5.0041  & (0.0464) &      0.0985  & (0.4163) &      5.0041  & (0.0464) &      0.0998  & (0.3989) \\
			1000  &      4.9949  & (0.0380) &      0.0979  & (0.3659) &      4.9948  & (0.0380) &      0.0978  & (0.3656) &      4.9948  & (0.0380) &      0.0984  & (0.3484) \\
			\midrule
			\multicolumn{12}{c}{b) True values: $\mu = 5$ and $\alpha = 0.3$}                             &  \\
			100   &      5.0010  & (0.1481) &      0.2704  & (0.4002) &      5.0008  & (0.1463) &      0.2674  & (0.3994) &      5.0022  & (0.1463) &      0.2894  & (0.3546) \\
			300   &      5.0041  & (0.0860) &      0.2883  & (0.2504) &      5.0043  & (0.0858) &      0.2874  & (0.2503) &      5.0044  & (0.0858) &      0.2956  & (0.2080) \\
			500   &      4.9989  & (0.0670) &      0.2937  & (0.1922) &      4.9990  & (0.0669) &      0.2932  & (0.1923) &      4.9989  & (0.0668) &      0.2980  & (0.1581) \\
			700   &      4.9952  & (0.0567) &      0.2958  & (0.1637) &      4.9950  & (0.0566) &      0.2953  & (0.1635) &      4.9950  & (0.0566) &      0.2982  & (0.1352) \\
			1000  &      4.9961  & (0.0474) &      0.2972  & (0.1357) &      4.9960  & (0.0474) &      0.2969  & (0.1356) &      4.9960  & (0.0474) &      0.2992  & (0.1094) \\
			\midrule
			\multicolumn{12}{c}{c) True values: $\mu = 5$ and $\alpha = 0.5$}                             &  \\
			100   &      5.0050  & (0.1937) &      0.4534  & (0.2538) &      5.0029  & (0.1914) &      0.4487  & (0.2560) &      5.0015  & (0.1904) &      0.4863  & (0.1938) \\
			300   &      5.0005  & (0.1102) &      0.4839  & (0.1494) &      5.0000  & (0.1098) &      0.4821  & (0.1501) &      5.0002  & (0.1097) &      0.4951  & (0.1098) \\
			500   &      5.0082  & (0.0848) &      0.4902  & (0.1153) &      5.0080  & (0.0846) &      0.4892  & (0.1156) &      5.0079  & (0.0845) &      0.4978  & (0.0840) \\
			700   &      5.0059  & (0.0734) &      0.4927  & (0.0980) &      5.0060  & (0.0733) &      0.4920  & (0.0982) &      5.0062  & (0.0732) &      0.4979  & (0.0706) \\
			1000  &      5.0017  & (0.0601) &      0.4945  & (0.0829) &      5.0018  & (0.0601) &      0.4940  & (0.0829) &      5.0020  & (0.0601) &      0.4987  & (0.0589) \\
			\midrule
			\multicolumn{12}{c}{d) True values: $\mu = 5$ and $\alpha = 0.7$}                             &  \\
			100   &      5.0223  & (0.2625) &      0.6401  & (0.1721) &      5.0149  & (0.2565) &      0.6330  & (0.1773) &      5.0129  & (0.2544) &      0.6847  & (0.1126) \\
			300   &      5.0007  & (0.1517) &      0.6801  & (0.0946) &      5.0008  & (0.1506) &      0.6776  & (0.0958) &      5.0004  & (0.1502) &      0.6954  & (0.0617) \\
			500   &      4.9826  & (0.1162) &      0.6848  & (0.0739) &      4.9821  & (0.1158) &      0.6834  & (0.0745) &      4.9810  & (0.1155) &      0.6958  & (0.0477) \\
			700   &      4.9935  & (0.0979) &      0.6892  & (0.0619) &      4.9924  & (0.0976) &      0.6882  & (0.0623) &      4.9916  & (0.0974) &      0.6973  & (0.0403) \\
			1000  &      5.0022  & (0.0828) &      0.6934  & (0.0512) &      5.0020  & (0.0827) &      0.6926  & (0.0513) &      5.0018  & (0.0826) &      0.6992  & (0.0335) \\
			\bottomrule
		\end{tabular}
	\end{adjustbox}
	\label{tab:5simulacao}
\end{table}

Figures \ref{fig:mu5} displays the boxplots of the estimates to the 5,000 Monte Carlo replicates, for each sample size and for each estimation method, for the parameters $\mu$ and $\alpha$. The horizontal dashed black line represents the true value of the parameters. As expected, we observe that as the sample size increases, the variance decreases and the estimates become more concentrated around the true value of the parameter. Also, that high $\alpha$ values impair the estimation of the $\mu$ parameter. As $\alpha$ increases, the variance of the $\mu$ estimates becomes greater. The first-panel line of Figure \ref{fig:mu5} corroborates what we state about ML method producing better results for $\alpha$ estimation. Note that ML estimates are closer to the true value of the parameter and as the sample size increases its variance decreases and became lower than the variance of CLS and YW estimates.

\begin{figure}[!htp]
	\centering
	\includegraphics[width=\textwidth]{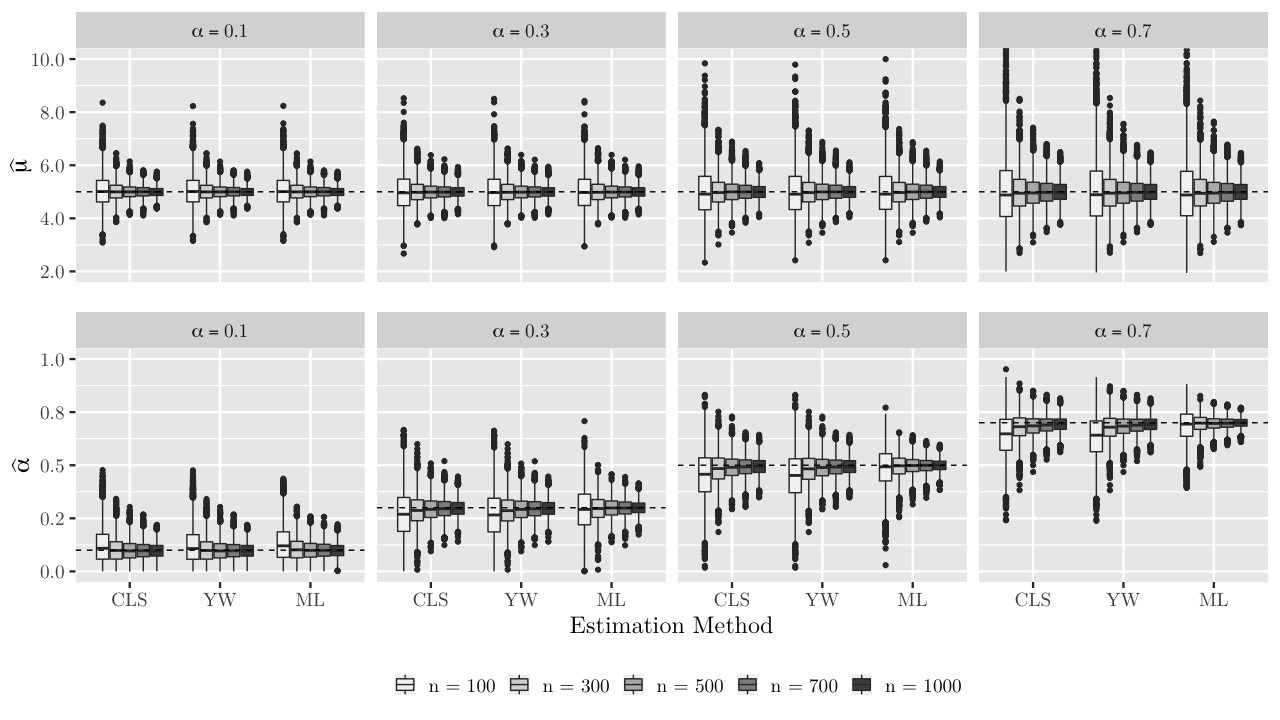}
	\caption{Boxplots of the estimates of the parameters $\mu$ and $\alpha$ based on the CLS, YW and ML estimation methods.}
	\label{fig:mu5}
\end{figure}

\section{Skin lesions data application}\label{chap:data}

In this section, we address the problem of analyzing the bovine skin lesions count time series dataset. The data, first presented in \citep{Jazi12}, consists of animal health laboratory submissions; provided by the Ministry of Agriculture and Forestry from New Zealand. The dataset has monthly submissions to animal health laboratories, from January $2003$ to December $2009$; a total of $n = 84$ observations from a region in New Zealand. The submissions are categorized in various ways: animal type, diseases, health symptoms, etc. Here, we analyze a monthly series giving the total number of bovine cases with skin lesions. 

This study is significant because of potential impact on public health and food supply. The cumulative number of skin lesion has been used as a measure of individual aggressiveness \citep{Turner06} and also as an indicator of the animals well-being in beef cattle facilities \citep{Platz09}. These factors impact animal behavior, which in turn affects the quality of the meat produced. Hence, to keep the number of lesions under control, being able to forecast an increase in the counts accurately, is fundamental in the management of livestock. 

To analyze this data, we shall use our proposed Geo-INAR(1) model. In addition, a comparison to the processes NGINAR(1), PINAR(1), and Poisson INARCH(1) is performed. PINAR(1) and NGINAR(1) processes are natural competitors of Geo-INAR(1) process. PINAR(1) process has marginals and innovations in the same Poisson family of distributions. Despite the parsimony of this model be an advantage, it is inadequate to deal with overdispersion. NGINAR(1) process has geometric marginals and a mixture of geometric distributions for the innovations, allowing to deal with the overdispersion problem. Poisson INARCH(1) belongs to the class of INGARCH models with an ARMA-like autocorrelation structure that can handle overdispersion.

\begin{remark}
For the Poisson INARCH(1) process, the marginal distribution, conditionally on the past, is Poisson; $X_t | X_{t-1}, X_{t-2},\ldots \sim \text{Poisson}(\beta + \alpha X_{t-1})$, where $\alpha \in (0,1)$ is the autocorrelation and the intercept $\beta>0$ is related to the marginal mean of the process in the following way: $E(X_t) = \mu = (1-\alpha)\beta$, $\forall t \in \Nset$.
\end{remark}

Figure \ref{fig:app2} shows the time series data, sample ACF and PACF. From these plots, we see that the autoregressive model of order 1 may be suitable for modeling the monthly count of skin lesions since there is a clear cut-off after lag 1 in the PACF. Moreover, the behavior of the series indicates that it may be mean stationary.

\begin{figure}[!ht]
	\centering
	\includegraphics[width=0.75\textwidth]{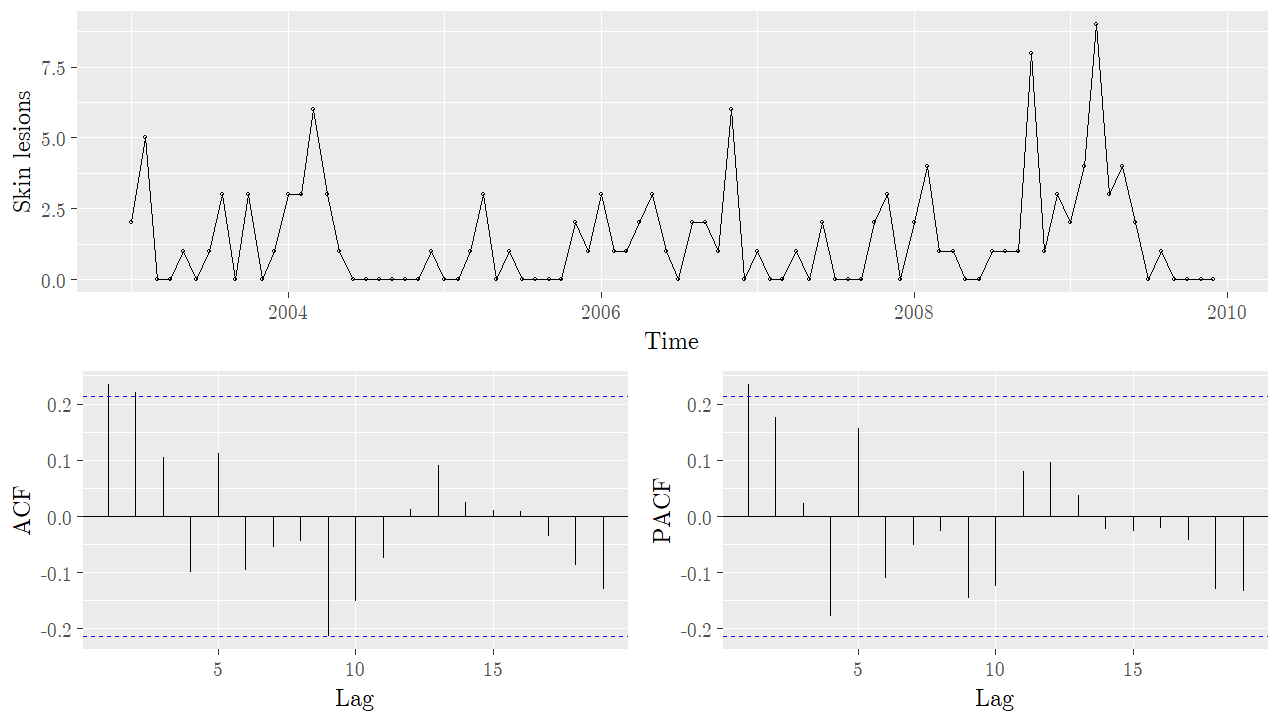}
	\caption{Plots of time series, ACF and PACF for the skin lesions dataset.} 
	\label{fig:app2}
\end{figure}

The sample mean, variance and autocorrelation of the data are $1.4286$, $3.3563$ and $0.2347$, respectively. The sample variance is about $2.5$ times the sample mean, what suggests overdispersion and thus the Poison distribution is not adequate and so the PINAR(1) process would be a poor choice to model the data. In order to properly check this affirmation, we perform an overdispersion test proposed by \citep{Schweer14}, where the test statistic is based on the empirical index of dispersion $\hat{I_d} := S^{2}/\bar{X}$, where $\bar{X}$ and $S^{2}$ are the sample mean and variance, respectively. The null hypothesis $\ho$: $X_1, \ldots, X_n$ \textit{stem from an equidispersed PINAR(1) process} against the alternative $\ha$: $X_1, \ldots, X_n$ \textit{stem from an overdispersed INAR(1) process} [as in the case of the Geo-INAR(1) process]. For the skin lesions data, we have that $\hat{I_d} = 2.3494$. The associated p-value is 
$$p_{\text{value}} = 1 - \Phi\left( \sqrt{\dfrac{n}{2}\dfrac{1 - \alpha^2}{1 + \alpha^2}}\cdot (\hat{I_d} - 1) \right) = 1.11\times 10^{-16},$$
where $\Phi(\cdot)$ denotes the distribution function of the standard normal distribution and if a hypothetical value for the dependence parameter $\alpha$ is not available, \citep{Schweer14} recommends to use a plug-in approach, i.e., to replace $\alpha$ by $\hat{\rho}_1 = 0.2347$. So, this value of the p-value, using any usual significance level (for instance $5\%$), leads to the rejection of the null hypothesis in favor of the alternative hypothesis that states that an overdispersed INAR(1) process is more adequate for modeling this dataset.

We fit to the data the Geo-INAR(1), NGINAR(1), PINAR(1) and Poisson INARCH(1) models. Despite the PINAR(1) model be rejected by the previous hypothesis test, we fit it for comparison purposes. Table \ref{tab:app2_estimates} lists the estimates of the parameters based on the ML estimation. Besides, we provide in this table the associated standard errors and confidence intervals of the parameters with significance level at $5\%$. The parameter $\mu$ gives the average number of bovines with skin lesions,  monthly. All models present similar values. For the Geo-INAR(1) model, we have $\hat{\mu}=1.4239$. The autocorrelation parameter, $\alpha$ indicates how strong is the dependence between the counts in one month to the next one. The Geo-INAR(1) and Poisson INARCH(1) models show a weak autocorrelation between observations, $\hat{\alpha} = 0.31$ and $0.34$, respectively; while the NGINAR(1) and PINAR(1) models present even a weaker dependence, around $0.17$.  

\begin{table}[!th]
\caption{Estimates of the parameters and its associated standard errors, 95$\%$ confidence intervals, AIC, and PMDA and PTP values for 1-step ahead forecast for the monthly count of skin lesions time series data.}
\centering
\begin{adjustbox}{max width=\textwidth}
\begin{tabular}{lcccccccc}
\toprule
\textbf{Model} & \textbf{Parameter} & \textbf{Estimate} & \textbf{Stand. Error} & \textbf{CI} & \textbf{AIC} & \textbf{PMAD} & \textbf{PTP} \\
\midrule
Geo-INAR & $\mu$    & 1.4239 & 0.2784 & (0.8782, 1.9696) & 266.10 & 1.000 & $25.0\%$  \\
		 & $\alpha$ & 0.3137 & 0.1178 & (0.0828, 0.5446) &        & \\
NGINAR   & $\mu$    & 1.4149 & 0.2423 & (0.9400, 1.8898) & 269.10 & 1.125 & $12.5\%$  \\
		 & $\alpha$ & 0.1717 & 0.1105 & (0.0000, 0.3882) &        & & \\
PINAR    & $\mu$    & 1.4264 & 0.1548 & (1.1230, 1.7298) & 298.20 & 0.875 & $25.0\%$  \\
	     & $\alpha$ & 0.1736 & 0.0682 & (0.0399, 0.3073) &        & & \\
INARCH   & $\mu$    & 1.4213 & 0.1963 & (1.0366, 1.8060) & 299.80 & 1.250 & $12.5\%$  \\
	     & $\alpha$ & 0.3391 & 0.0885 & (0.1656, 0.5125) &        & & \\
\bottomrule
\end{tabular}
\end{adjustbox}
\label{tab:app2_estimates}
\end{table}

To perform a more accurate comparison between the models, Table \ref{tab:app2_qties} presents empirical and estimated quantities - plugged-in the ML estimates. Since the Geo-INAR(1) and NGINAR(1) models have geometric marginals, quantities based only on the first moments are insufficient to make a good comparison. Thus, we use mixed moments up to order 4 of an INAR(1) process given by \citep{Schweer14}. Empirically, the mixed moments up to order 4 are defined through the following notation
$$\mu(s_1, \ldots, s_{r-1}) := E(X_t \cdot X_{t+s_1}\cdots X_{t+s_{r-1}}), 0 \leq s_1 \leq \ldots \leq s_{r-1} \text{ and } r \in \Nset.$$
So the case $r = 1$ corresponds to the marginal mean $\mu_X = E(X_t)$. In the case of a stationary INAR(1) process \citep{Schweer14} prove that the first and second higher-order moments are given by

$\mu(k) = \sigma_{X}^{2} \alpha^{k} + \mu_{X}^{2},$ and

$\mu(k, l) = (\bar{\mu}_{X, 3} - \sigma_{X}^{2})\alpha^{l+k} +
(1 + \mu_{X})\sigma_{X}^{2} \alpha^{l} + \mu_{X}\sigma_{X}^{2}(\alpha^{l-k} + \alpha^{k}) + \mu_{X}^{3},$

\noindent respectively, for any $0 \leq k \leq l$, where $\bar{\mu}_{X, r} := E[(X - \mu_X)^{r}]$ denotes the central moments of $X$, and the innovations ${\varepsilon_t}$ have existing moments $\mu_{\varepsilon, r} := E(\varepsilon_{t}^{r})$, for $r \leq 4$.

\begin{table}[!h]
	\caption{Comparison among the models based on empirical and estimated quantities for the monthly count of skin lesions time series data.}
	\centering
	\begin{tabular}{lrrrr}
		\toprule
		\textbf{Quantity} & \textbf{Empirical} & \textbf{Geo-INAR} & \textbf{NGINAR} & \textbf{PINAR} \\
		\midrule
		$\kappa_1$        &  1.4286            &  1.4239           &  1.4149         & 1.4264 \\
		$\kappa_2$        &  3.3563            &  3.4514           &  3.4168         & 1.4264 \\
		skew.             &  1.8378            &  2.0712           &  2.0719         & 0.8373 \\
		kurt.             &  6.8999            &  6.2897           &  6.2927         & 0.7011 \\
		$I_d$             &  2.3494            &  2.4239           &  2.4149         & 1.0000 \\
		$p_0$             &  0.4048            &  0.4126           &  0.4141         & 0.2402 \\
		$\mu(1)$          &  2.8434            &  3.1102           &  2.5886         & 2.2822 \\
		$\mu(2)$          &  2.7683            &  2.3671           &  2.1027         & 2.0776 \\
		$\mu(1, 1)$       & 12.1205            & 12.9346           & 10.1989         & 5.8908 \\
		$\mu(1,2)$        &  6.9756            &  7.0969           &  4.7849         & 3.7129 \\
		\bottomrule
	\end{tabular}
	\label{tab:app2_qties}
\end{table}

In Table \ref{tab:app2_qties} consider the following quantities: mean, variance, skewness, kurtosis, index of dispersion (ratio of the variance to the mean), probability of zero, first higher-order moment with lag 1, first higher-order moment with lag 2, and second high-order moment with lag 1, denoted by $\kappa_1$, $\kappa_2$, skew., kurt., $I_d$, $p_0$, $\mu(1)$, $\mu(2)$, and $\mu(1,1)$, respectively. Note that the Geo-INAR(1) model estimates quantities closer to the empirical ones, specially in terms of high-order moments, quantities related to the correlation structure of the process. Therefore, our process not only has better adherence to the sample distribution but also has a superior performance in capturing the correlation structure in the data. As expected, the PINAR(1) model presents poor results, reinforcing its inadequacy to modeling this dataset. As regards Poisson INARCH(1) model, it is cumbersome to obtain the high-order moments and, to the best of our knowledge, have not been provided in previous papers.

The PIT histogram is a graphical representation of the probability integral transform \citep{czado09}. Besides being a tool for evaluating the goodness-of-fit of a model, the PIT histogram is also used to check the predictive performance since it considers not only conditional moments but the complete conditional distribution. If the fitted model is adequate, the PIT histogram mimics the shape of a uniform histogram. U-shaped and inverse-U shaped histograms indicate underdispersion and overdispersion, respectively, from the uniform distribution. Figure \ref{fig:PIT2} presents the PIT histogram for all four models and corroborate the fact that the PINAR(1) and the Poisson INARCH(1) cannot handle with the overdispersion presented in the data. Another way to compare models is by calculating the Akaike Information Criterion (AIC); the smaller the AIC of a model, the better the model is. Table \ref{tab:app2_estimates} shows that the Geo-INAR(1) model presents smaller values of the AIC.

\begin{figure}[!ht]
	\centering
	\includegraphics[width=0.75\textwidth]{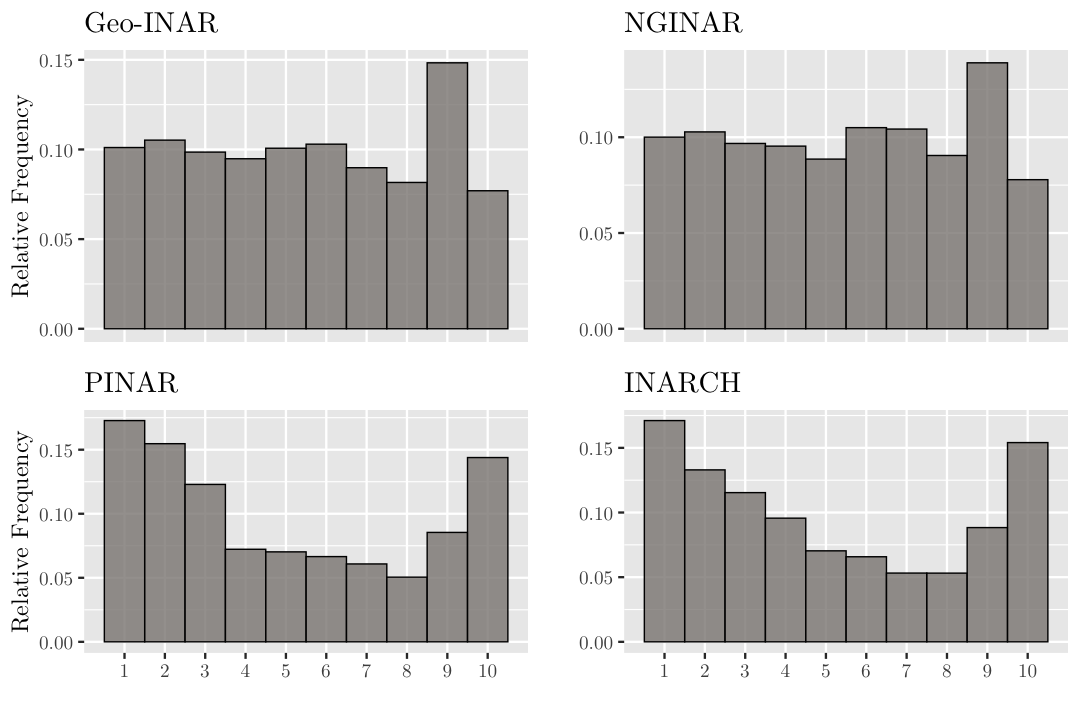}
	\caption{PIT histograms for the monthly count of skin lesions time series for different models.} 
	\label{fig:PIT2}
\end{figure}

So far, we may conclude that the Geo-INAR(1) and NGINAR(1) models are more appropriate choices for modeling the dataset. However, the Geo-INAR(1) presents smaller AIC and has more adherence to the empirical data than the NGINAR(1). So, we now proceed to discuss the goodness-of-fit of the Geo-INAR(1) model based on the residuals $R_t := X_t - \hat{E}(X_t|X_{t-1})$ and on the jumps $J_t := X_t - X_{t-1}$, for $t = 2, \ldots, n$. Note that $E(J_t) = 0$ and $Var(J_t) = 2\mu(1+\mu)(1-\alpha)$. In Figure \ref{fig:resJump2}, we present plots of the sample autocorrelation function of the residuals and the jumps against time (Shewhart control chart) with $\pm 3\sigma_J$ limits chosen as the benchmark chart as proposed in \citep{Weib09}, where $\sigma_J = \sqrt{Var(J_t)} = 2.1766$. These plots indicate that the residuals $R_2, \ldots, R_n$ are not correlated, hence our model seems to have captured well the dependence of the time series and that there is not a particular point causing a huge impact in the model. From the Shewhart control chart around 97\% of the points are within the control limits. 
Additionally, we perform the Ljung-Box test to ensure the independence of the residuals \citep{boxjenkins}. The associated p-value, for lag $1$, is $0.2017$, thus we accept the null hypothesis that the residuals are independently distributed, using any usual significance level. The same is valid for higher-order lags, as shown in the third panel of Figure \ref{fig:resJump2}.

\begin{figure}[ht]
	\centering
	\includegraphics[width=\textwidth]{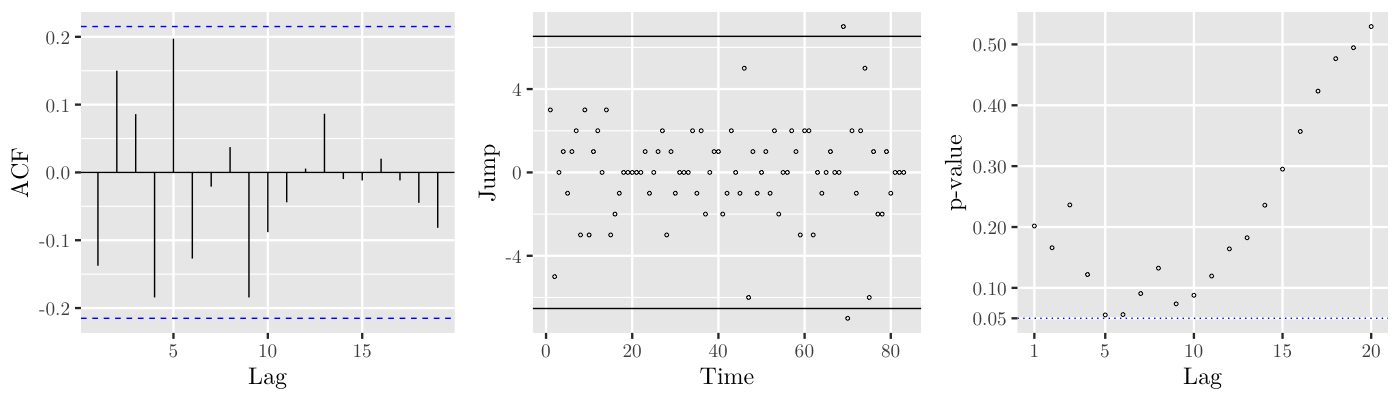}
	\caption{Plots of the sample ACF of the residuals, jumps against time, and Ljung-Box p-values for the skin lesions dataset.} 
	\label{fig:resJump2}
\end{figure}

Also, we calculate the standardized Pearson residuals, 
$$r_t := R_t / \sqrt{\widehat{Var}(X_t|X_{t-1})},\; \text{ for } t = 2, \ldots, n,$$ 
to check if the Geo-INAR(1) model has well captured the overdispersion. According to \citep{Harvey89}, a sample variance of the residuals greater than 1 indicates overdispersion with respect to the model that is being fitted. Since the empirical variance of the residuals is $0.9606$, there is evidence that our Geo-INAR(1) process captured well the overdispersion of the data. On the other hand, the NGINAR(1) process has an empirical variance of the residuals equals to $0.9143$, a worse performance than our model.

To verify the prediction performance of each model, we separated the dataset into a training dataset ($n^{\prime} = 76$ observations) and a test dataset ($m=8$ observations). In the training dataset, we computed the ML estimators using each of the four processes while the test dataset, was used to calculate the $1$-step ahead conditional median as our point forecast. To compare the results we use the predicted mean absolute deviation (PMAD) and the percentage of true prediction (PTP) given, respectively, by
\begin{align*}
\text{PMAD} &= \frac{1}{m}\sum_{t = {n^{\prime}}}^{{n^{\prime}}+m-1}{|X_{t+1} - \hat{X}_{t+1}|},\; \text{ and}\\ 
\text{PTP} &= \frac{1}{m}\sum_{t={n^{\prime}}}^{{n^{\prime}}+m-1}{\ind_{\{X_{t+1} = \hat{X}_{t+1}\}}}\times 100\%,
\end{align*}
where $\hat{X}_{t+1}$ is the predicted value and $\ind_{\{\cdot\}}$ is the indicator function \citep{maiti14, Maiti15}. Table \ref{tab:app2_estimates} shows the comparison metrics; the Geo-INAR(1) process is the best forecaster.

After all, the Geo-INAR(1) model not only has better adherence to the sample distribution of the counts of skin lesion but also has higher forecast performance. The model is an optimal choice to track livestock behavior, helping to keep the number of lesions within an acceptable level, maximizing the quality of the meat produced.

\section*{Acknowledgement}
M. B. Guerrero and H. Ombao would like to acknowledge support for their research support from KAUST. W. Barreto-Souza would like to thank the financial support from Conselho Nacional de Desenvolvimento Científico e Tecnológico (CNPq-Brazil, grant number 305543/2018-0) and Fundação de Amparo à Pesquisa do Estado de Minas Gerais (FAPEMIG-Brazil, grant number APQ-00786-14).

\bibliographystyle{unsrt}  


\appendix
\section{Proofs of propositions}\label{chap:demos}

\noindent \textbf{Note:} Main proofs of the paper.

\propsumzmgcondlabel*
\begin{proof} 
    Note that
    \begin{align*}
    P \left( \bm{\theta}\star X = k | X = x  \right) &= P \left( \sum_{i=1}^{X}{G_{i}} = k \bigg | X = x  \right) \\
	&= P \left( \sum_{i=1}^{x}{G_{i}} = k \right)\\
	&= P \left( S_{x} = k \right).
    \end{align*}
	
	We want to find the pmf of $S_x = \sum\limits_{i=1}^{x}{G_i}$, the sum of $x$ iid ZMG random variables. One way to do this is to compute the pgf of $S_x$.
	\begin{align*}
	\Psi_{S_x}(s) &= E(s^{S_x}) = E\left(s^{ \sum_{i=1}^{x}{G_i} }\right) = E\left( \prod_{i=1}^{x}s^{G_i} \right) (since\,\, \{G_i\} \; \text{is iid} \; \forall i)\\
	&= \prod_{i=1}^{x}E(s^{G}) = (\Psi_{G}(s))^{x} = \left(\dfrac{ 1+\pi\mu_{\varepsilon}(1-s) }{ 1+\mu_{\varepsilon}(1-s) }\right)^{x}.\numberthis 
	\label{eq:gpfsx}
	\end{align*}
	
	Now, rearranging $\Psi_{S_{x}}(s)$ given by Eq. \eqref{eq:gpfsx} in a convenient way and following the steps used by Kolev et al. (2000) in their Proposition $4.1$, with the right reparameterization, we have that
	\begin{align*}
	\Psi_{S_x}(s) &= \left(\dfrac{ 1+\pi\mu_{\varepsilon}(1-s) }{ 1+\mu_{\varepsilon}(1-s) }\right)^{x}\\ 
	&= \left(\frac{1+\pi\mu_{\varepsilon}}{1+\mu_{\varepsilon}}\right)^{x}\left(1 - \frac{(1-\pi)\mu_{\varepsilon}}{(1+\mu_{\varepsilon})(1+\pi\mu_{\varepsilon})}\frac{s}{1-s\frac{\pi\mu_{\varepsilon}}{1+\pi\mu_{\varepsilon}}}\right)^{-x}.
	\end{align*}
	
	Then, we use the fact that $(1-z)^{-x} = \displaystyle\sum_{i=0}^{\infty}{\binom{i+x-1}{i}z^{i}}$, for $|z| < 1$ and $x \geq 1$. Here, we have that 
    $$z = \frac{(1-\pi)\mu_{\varepsilon}}{(1+\mu_{\varepsilon})(1+\pi\mu_{\varepsilon})}\frac{s}{1 - s{\frac{\pi\mu_{\varepsilon}}{1+\pi\mu_{\varepsilon}}}},$$
    which imposes a constraint to the $\alpha$ parameter in terms of $\mu$: $\alpha < \dfrac{\mu}{1+\mu}$. 
    
    Note that the NGINAR(1) process has the same constraint.
	
	Hence, we get that
	\begin{align*}
	&\Psi_{S_{x}}(s) = \left(\frac{1+\pi\mu_{\varepsilon}}{1+\mu_{\varepsilon}}\right)^{x}\sum\limits_{i=0}^{\infty}{\binom{i+x-1}{i}\left( \frac{(1-\pi)\mu_{\varepsilon}}{(1+\mu_{\varepsilon})(1+\pi\mu_{\varepsilon})}\frac{s}{1-s\frac{\pi\mu_{\varepsilon}}{1+\pi\mu_{\varepsilon}}}  \right)^{i}}\\[10pt]
	&= \left(\frac{1+\pi\mu_{\varepsilon}}{1+\mu_{\varepsilon}}\right)^{x} + \left(\frac{1+\pi\mu_{\varepsilon}}{1+\mu_{\varepsilon}}\right)^{x}\times \\
	&\phantom{=} \sum\limits_{i=1}^{\infty}{s^{i}\binom{i+x-1}{i} \left(\frac{(1-\pi)\mu_{\varepsilon}}{1+\mu_{\varepsilon}}\frac{1}{1+\pi\mu_{\varepsilon}}\right)^{i}\left(1 - s\frac{\pi\mu_{\varepsilon}}{1+\pi\mu_{\varepsilon}}\right)^{-i}} \\[10pt]
	&= \left(\frac{1+\pi\mu_{\varepsilon}}{1+\mu_{\varepsilon}}\right)^{x} + \left(\frac{1+\pi\mu_{\varepsilon}}{1+\mu_{\varepsilon}}\right)^{x}\times \\
	&\phantom{=} \sum\limits_{i=1}^{\infty}{s^{i}\binom{i+x-1}{i} \left(\frac{(1-\pi)\mu_{\varepsilon}}{1+\mu_{\varepsilon}}\frac{1}{1+\pi\mu_{\varepsilon}}\right)^{i}}\sum\limits_{j=0}^{\infty}{\binom{j+i-1}{j}\left(s\frac{\pi\mu_{\varepsilon}}{1+\pi\mu_{\varepsilon}}\right)^{j}} \\[10pt]
	&= \left(\frac{1+\pi\mu_{\varepsilon}}{1+\mu_{\varepsilon}}\right)^{x} + \left(\frac{1+\pi\mu_{\varepsilon}}{1+\mu_{\varepsilon}}\right)^{x}\times \\
	&\phantom{=} \sum\limits_{i=1}^{\infty}\sum\limits_{j=0}^{\infty}{s^{i+j}\binom{i+x-1}{i}\binom{j+i-1}{j}\left(\frac{(1-\pi)\mu_{\varepsilon}}{1+\mu_{\varepsilon}}\frac{1}{1+\pi\mu_{\varepsilon}}\right)^{i}\left(\frac{\pi\mu_{\varepsilon}}{1+\pi\mu_{\varepsilon}}\right)^{j}} \\[10pt]
	&\stackrel{(I)}{=} \left(\frac{1+\pi\mu_{\varepsilon}}{1+\mu_{\varepsilon}}\right)^{x} + \left(\frac{1+\pi\mu_{\varepsilon}}{1+\mu_{\varepsilon}}\right)^{x}\times \\
	&\phantom{=} \sum\limits_{i=1}^{\infty}\sum\limits_{j=0}^{\infty}{s^{i+j}\binom{i+x-1}{i}\binom{j+i-1}{i-1}\left(\frac{(1-\pi)\mu_{\varepsilon}}{1+\mu_{\varepsilon}}\frac{1}{1+\pi\mu_{\varepsilon}}\right)^{i}\left(\frac{\pi\mu_{\varepsilon}}{1+\pi\mu_{\varepsilon}}\right)^{j}} \\[10pt]
	&\stackrel{(II)}{=} \left(\frac{1+\pi\mu_{\varepsilon}}{1+\mu_{\varepsilon}}\right)^{x} + \left(\frac{1+\pi\mu_{\varepsilon}}{1+\mu_{\varepsilon}}\right)^{x}\times \\
	&\phantom{=} \sum\limits_{k=1}^{\infty}s^{k}\sum\limits_{i=1}^{k}{\binom{i+x-1}{i}\binom{k-1}{i-1}\left(\frac{(1-\pi)\mu_{\varepsilon}}{1+\mu_{\varepsilon}}\frac{1}{1+\pi\mu_{\varepsilon}}\right)^{i}\left(\frac{\pi\mu_{\varepsilon}}{1+\pi\mu_{\varepsilon}}\right)^{k-i}}\\[10pt]
	&= \left(\frac{1+\pi\mu_{\varepsilon}}{1+\mu_{\varepsilon}}\right)^{x} + \left(\frac{1+\pi\mu_{\varepsilon}}{1+\mu_{\varepsilon}}\right)^{x}\times \\
	&\phantom{=} \sum\limits_{k=1}^{\infty}s^{k}\sum\limits_{i=1}^{k}{\binom{k-1}{i-1}\binom{i+x-1}{i}\left(\frac{(1-\pi)\mu_{\varepsilon}}{1+\mu_{\varepsilon}}\right)^{i}\left(\frac{1}{1+\pi\mu_{\varepsilon}}\right)^{i}\left(\frac{\pi\mu_{\varepsilon}}{1+\pi\mu_{\varepsilon}}\right)^{k-i}} \\[10pt]
	&= \pi_{\star}^{x} + \pi_{\star}^{x}\sum\limits_{k=1}^{\infty}s^{k}\sum\limits_{i=1}^{k}{ \binom{i+x-1}{i}(1-\pi_{\star})^{i}\binom{k-1}{i-1}p^{k-i}(1-p)^{i} } \\[10pt]  
	&= \pi_{\star}^{x} + \sum\limits_{k=1}^{\infty}s^{k}P(S_{x} = k),
	\end{align*}
	where for $k =1$, $2$, $\ldots$, 
	$$P(S_{x} = k) = \displaystyle{
		\sum\limits_{i=1}^{k}{ \binom{x+i-1}{i}\pi_{\star}^{x}(1-\pi_{\star})^{i} \binom{k-1}{i-1}p^{k-i}(1-p)^{i} } }.$$
	
	As we can see, the expression here obtained to the pmf of the ZMNB distribution is simpler than that one obtained by Kolev et al. (2000).
	
	\noindent \textbf{Notes:}
	\begin{enumerate}
		\item [(I)] In this step we use the fact that $\displaystyle{\binom{i+j-1}{j} = \binom{j+i-1}{i-1}}$.
		\item [(II)] Use the following result:
		\begin{lemma} $\sum\limits_{i=1}^{\infty}{\sum\limits_{j=0}^{\infty}{s^{i+j}a_{i}b_{j}}}=\sum\limits_{k=1}^{\infty}{s^{k}c_{k}}$, where ${c_{k}=\sum\limits_{l=1}^{k}{a_{l}b_{k-l}}}.$
		\end{lemma}
	\end{enumerate}
	
	To obtain the mean and the variance of $S_x$, note that
	$$E(S_x) = x \mu_\varepsilon(1-\pi) = x E(G)$$
	and
	$$Var(S_x) = x \mu_\varepsilon(1-\pi)\left[1 + \mu_\varepsilon(1+\pi)\right] = x Var(G).$$
\end{proof}

\proptransitionlabel*
\begin{proof}
	\begin{enumerate}
		\item [i.] From state $0$ to state $j$ the proof is straightforward.
		
		From state $i \geq 1$ to state $j$, it follows that
		\begin{align*}
		p_{ij} &= P(X_t = j | X_{t-1} = i) = \frac{P(X_t = j, X_{t-1}=i)}{P(X_{t-1}=i)} \\[7pt]
		&= \frac{P(\bm{\theta}\star X_{t-1} + \varepsilon_t = j, X_{t-1}=i)}{P(X_{t-1}=i)}
		= \frac{P(\varepsilon_t = j - \bm{\theta}\star X_{t-1}, X_{t-1}=i)}{P(X_{t-1}=i)}\\[7pt]
		&= \frac{P(\varepsilon_t = j - \sum_{k=1}^{X_{t-1}}{G_{k}}, X_{t-1}=i)}{P(X_{t-1}=i)}
		= \frac{P(\varepsilon_t = j - \sum_{k=1}^{i}{G_{k}}, X_{t-1}=i)}{P(X_{t-1}=i)}\\[7pt]
		&= \frac{P(\varepsilon_t = j - \sum_{k=1}^{i}{G_{k}})P(X_{t-1}=i)}{\Pr(X_{t-1}=i)}
		= P\left(\varepsilon_t = j - \sum_{k=1}^{i}{G_{k}}\right)\\[7pt]
		&= \sum_{m=0}^{j}{P\left(\varepsilon_t = j - \sum_{k=1}^{i}{G_{k}}, \sum_{k=1}^{i}{G_{k}}=m\right)}\\
		&= \sum_{m=0}^{j}{P\left(\varepsilon_t = j - m, \sum_{k=1}^{i}{G_{k}}=m\right)}\\[7pt]
		&= \sum_{m=0}^{j}{P(\varepsilon_t = j - m)P\left(\sum_{k=1}^{i}{G_{k}}=m\right)}\\
		&= \sum_{m=0}^{j}{\frac{\mu_{\varepsilon}^{j-m}}{(1+\mu_{\varepsilon})^{(j-m)+1}} P\left(\sum_{k=1}^{i}{G_{k}}=m\right)}\\[7pt]
		&= P(\varepsilon_t = j)\sum_{m=0}^{j}{\left( \frac{1+\mu_{\varepsilon}}{\mu_{\varepsilon}} \right)^{m}}P\left(S_{i}=m\right),\; S_{i} \sim \text{ZMNB}(\pi, \mu_\varepsilon, i),\; i \geq 1\\[7pt]
		&= p_{\varepsilon}(j)\left[ P(S_i = 0) + \sum_{m=1}^{j}\left(1 + \frac{1}{\mu_\varepsilon}\right)^{m}P(S_i = m) \right]\\[7pt]
		&= p_{\varepsilon}(j)\displaystyle{ \left[
			\pi_{\star}^{i} + \sum_{m=1}^{j}{\sum_{l=1}^{m}{ \left(1 + \frac{1}{\mu_\varepsilon}\right)^{m}\mbox{A}_{l}^{i}(\pi_{\star})\mbox{B}_{l}^{m}(p)}}\right]},\; \text{ where}
		\end{align*}
		$$\mbox{A}_{l}^{i}(\pi_{\star}) := \displaystyle{\binom{i+l-1}{l}\pi_{\star}^{i}(1-\pi_\star)^{l}}\;
\text{ and }\; 
\mbox{B}_{l}^{m}(p) := \displaystyle{\binom{m-1}{l-1}p^{m-l}(1-p)^{l}}.$$

		\item [ii.] $\begin{aligned}[t] \phantom{=} E(X_{t+1}|X_t) &= E(\bm{\theta}\star X_t+\varepsilon_{t+1}|X_t)\\
		&= E(\bm{\theta}\star X_t|X_t) + E(\varepsilon_{t+1}|X_t)\\
		&= E(G)X_{t} + E(\varepsilon)\\
		&= \alpha X_t + (1-\alpha)\mu.\end{aligned}$
		
		\item [iii.] $\begin{aligned}[t] &\phantom{=} Var(X_{t+1}|X_t) = E\left(X_{t+1}^{2}|X_t\right) - E^{2}\left(X_{t+1}|X_t\right)\\[13pt]
		&\phantom{=}\quad = E\left((\bm{\theta}\star X_{t}+\varepsilon_{t+1})^{2}|X_t\right) - \underbrace{\left( E(G)X_t + E(\varepsilon_{t+1}) \right)^{2}}_{\zeta}\\
		&\phantom{=}\quad = E\left( (\bm{\theta}\star X_t)^{2}|X_t \right) + E\left( 2(\bm{\theta}\star X_t)\varepsilon|X_t \right) + E(\varepsilon^{2}) - \zeta\\[13pt]
		&\phantom{=}\quad = X_t E(G^{2}) + (X^{2}-X)E^{2}(G) + 2X_tE(G)E(\varepsilon) + Var(\varepsilon) + E^{2}(\varepsilon) - \zeta\\[13pt]
		&\phantom{=}\quad = X_t Var(G) + \underbrace{\left( E(G)X_t + E(\varepsilon) \right)^{2}}_{\zeta} + Var(\varepsilon) - \zeta\\
		&\phantom{=}\quad = Var(G)X_t + Var(\varepsilon)\\[13pt]
		&\phantom{=}\quad = \left[ (1+2\mu)(1-\alpha)\alpha \right]X_t + \sigma_{\varepsilon}^{2}.\end{aligned}$
		
		\item [iv.] It follows from routine calculations from INAR processes.
	\end{enumerate}
\end{proof} 

\section{Supplementary Material}\label{chap:sup}

\noindent \textbf{Note:} Remaining proofs of the paper.

\propgenerallabel*
\begin{proof} In what follows, we use the definition of the $\star$ operator and basics results of conditional expectation.
	
	\begin{enumerate}
		\item[i.] $\begin{aligned}[t] E(\bm{\theta} \star X) &= E\left[ E\left(\bm{\theta} \star X|X \right) \right] = E\left[ E\left(\sum_{i=1}^{X}{G_i}\bigg|X\right) \right]\\
		&= E\left[ \sum_{i=1}^{X}{E(G_{i}|X)} \right] \quad (\mbox{since}\,\, G_{i} \perp X, \forall i = 1, 2, \ldots)\\
		&= E\left( \sum_{i=1}^{X}{E(G_i)} \right) \quad (\mbox{the sequence}\,\, \{G_i\}_{i=1}^\infty \text{ is iid})\\
		&= E\left( \sum_{i=1}^{X}{E(G)} \right)\\
		&= E\left( X E(G) \right)\\
		&= E(G)E(X).\end{aligned}$
		
		\item [ii.] Here some additional algebraic manipulations are required. We have that
		\begin{flalign*} E((\bm{\theta} \star X)^{2}) &= E\left[ E \left( \left(\bm{\theta} \star X\right)^{2}|X \right)  \right]
		= E\left[ E\left( \left( \sum_{i=1}^{X}{G_i} \right)\left( \sum_{j=1}^{X}{G_j} \right) \bigg|X\right)    \right]\\
		&= E\left[ E\left( \left(\sum_{i=1}^{X}{G_{i}^{2}} + \sum_{i=1}^{X}\sum_{\stackrel{j=1}{i \neq j}}^{X}{G_{i}G_{j}}\right) \bigg|X \right)   \right]\\
		&= E\left[ E\left(\sum_{i=1}^{X}{G_{i}^{2}}\bigg|X \right) \right] 
		+ E\left[ E\left( \sum_{i=1}^{X}\sum_{\stackrel{j=1}{i \neq j}}^{X}{G_{i}G_{j}}\bigg|X \right) \right]\\
		&= E\left[ \sum_{i=1}^{X}{E\left(G_{i}^{2}|X \right)} \right] 
		+ E\left[ \sum_{i=1}^{X}\sum_{\stackrel{j=1}{i \neq j}}^{X}{E\left( G_{i}G_{j}|X \right)} \right]\\
		&= E\left[ \sum_{i=1}^{X}{E\left(G_{i}^{2} \right)} \right] 
		+ E\left[ \sum_{i=1}^{X}\sum_{\stackrel{j=1}{i \neq j}}^{X}{E\left( G_{i}G_{j} \right)} \right]\\
		&= E\left[ \sum_{i=1}^{X}{E\left(G^{2} \right)} \right] 
		+ E\left[ \sum_{i=1}^{X}\sum_{\stackrel{j=1}{i \neq j}}^{X}{E^{2}(G)} \right]\\
		&= E\left[X E(G^2)\right] + E\left[ (X^2 - X)E^2(G) \right]\\
		&= E(X)E(G^2) + (E(X^2) - E(X))E^2(G)\\
		&= (E(G^2) - E^2(G))E(X) + E^2(G)E(X^2)\\
		&= Var(G)E(X) + E^{2}(G)E(X^{2}).
		\end{flalign*}
		
		\item [iii.] $\begin{aligned}[t] E\left((\bm{\theta} \star X)Y\right) &= E\left[ E\left( \left( \sum_{i=1}^{X}{G_i} \right)Y\bigg|X,Y \right) \right] = E\left[ Y \sum_{i=1}^{X}{E(G_{i}|X,Y)} \right]\\
		&= E\left[ Y \sum_{i=1}^{X}{E(G_{i})} \right] = E\left[ Y \sum_{i=1}^{X}{E(G)} \right] = E\left[ XY E(G) \right]\\
		&= E(G)E(XY).\end{aligned}$
		
		\item [iv.] Immediate. It follows from the application of (i) and (ii) in the K\"{o}ening formula.
		
		\item [v.] Immediate. It follows directly by using (i) and (iii).
	\end{enumerate}
\end{proof}

\begin{lemma}\label{lemma.comp} 
	Let $G_{i} \sim \text{ZMG}(\pi_i, {\mu_{\varepsilon}}_i)$, $i=1, 2, \ldots, h$, $h \geq 1$, with $\pi_i = 1- \alpha_i/{\mu_{\varepsilon}}_i$ and ${\mu_{\varepsilon}}_i = (1-\alpha_i)\mu$. Here, $\mu > 0$ and $0<\alpha_i < 1$, $\forall i$. Then,
	\begin{align}\label{eq:comppgfs}
	\Psi_{G_{h}}&\left( \Psi_{G_{h-1}}\left( \ldots \Psi_{G_{2}} \left(\Psi_{G_{1}}(s)\right)\right)\right)=\nonumber \\
	&\displaystyle{\frac{1 + \left[1 - \frac{\prod_{i=1}^{h}\alpha_i}{\left(1-\prod_{i=1}^{h}\alpha_i\right)\mu}\right]\left(1-\prod_{i=1}^{h}\alpha_i\right)\mu\left(1-s\right)}{1 + \left(1-\prod_{i=1}^{h}\alpha_i\right)\mu\left(1-s\right) }},
	\end{align}
	which is a pgf of a random variable with distribution
	$$\text{ZMG}\left(1-\frac{\prod_{i=1}^{h}\alpha_i}{\left(1-\prod_{i=1}^{h}\alpha_1\right)\mu}, \left(1-\prod_{i=1}^{h}\alpha_1\right)\mu \right).$$
	
	The case where all the $h$ random variable $G_i$'s have the same parameters $\pi = 1- \alpha/{\mu_{\varepsilon}}$ and ${\mu_{\varepsilon}} = (1-\alpha)\mu$ leads to a pgf of a random variable with distribution ZMG$\left(1-\frac{\alpha^h}{\left(1-\alpha^h\right)\mu}, \left(1-\alpha^h\right)\mu \right)$.
\end{lemma}
\begin{proof}
	We use mathematical induction:
	\begin{itemize}
		\item Base case: $h=2$.
		
		Note that
		$$\Psi_{G_{1}}(s) = \frac{1 + \left( 1 - \alpha_{1}/{\mu_{\varepsilon}}_{1} \right){\mu_{\varepsilon}}_{1}(1-s)}{1+{\mu_{\varepsilon}}_{1}(1-s)} \Rightarrow
		1 - \Psi_{G_{1}}(s) = \frac{\alpha_{1}(1-s)}{1 + {\mu_{\varepsilon}}_{1}(1-s)}.$$
		Replacing this result into $\Psi_{G_{2}}\left(\Psi_{G_{1}}(s)\right)$ we have that
		\begin{align*}
		\Psi_{G_{2}}\left(\Psi_{G_{1}}(s)\right) &= \frac{1 + \left( 1 - \alpha_{2}/{\mu_{\varepsilon}}_{2} \right){\mu_{\varepsilon}}_{2}\left(1-\Psi_{G_{1}}(s)\right)}{1+{\mu_{\varepsilon}}_{2}\left(1-\Psi_{G_{1}}(s)\right)}\\
		&=\frac{1 + \left[(1 - \alpha_{2}/{\mu_{\varepsilon}}_{2}){\mu_{\varepsilon}}_{2}\right]\left[\frac{\alpha_{1}(1-s)}{1 + {\mu_{\varepsilon}}_{1}(1-s)}\right]}{\frac{1+(1-\alpha_{1}\alpha_{2})\mu(1-s)}{1+{\mu_{\varepsilon}}_{1}(1-s)}}\\
		&=\frac{1 + {\mu_{\varepsilon}}_{1}(1-s) + \left[{\mu_{\varepsilon}}_{2} - \alpha_{2}\right]\alpha_{1}(1-s)}{1+(1-\alpha_{1}\alpha_{2})\mu(1-s)}\\
		&=\frac{1 + \left[1 - \frac{\alpha_{1}\alpha_{2}}{(1-\alpha_{1}\alpha_{2})\mu}\right](1-\alpha_{1}\alpha_{2})\mu(1-s)}{1 + (1-\alpha_{1}\alpha_{2})\mu(1-s)},
		\end{align*}
		which is a pgf of a random variable with distribution 
		$$\text{ZMG}\left( 1 - \frac{\alpha_{1}\alpha_{2}}{(1-\alpha_{1}\alpha_{2})\mu}, (1-\alpha_{1}\alpha_{2})\mu   \right).$$
		
		\item Step case: $h \in \Nset$.
		
		Suppose the result holds for $h - 1 \in \Nset$, $h > 2$, and evaluate the expression for $h \in \Nset$.
		
		By the induction hypothesis we have that
		$$\Psi_{G_{h-1}}\left( \ldots \Psi_{G_{2}} \left(\Psi_{G_{1}}(s)\right)\right)=
		\displaystyle{\frac{1 + \left[1 - \frac{\prod_{i=1}^{h-1}\alpha_i}{\left(1-\prod_{i=1}^{h-1}\alpha_i\right)\mu}\right]\left(1-\prod_{i=1}^{h-1}\alpha_i\right)\mu\left(1-s\right)}{1 + \left(1-\prod_{i=1}^{h-1}\alpha_i\right)\mu\left(1-s\right) }}.$$
		
		Hence,
		$$1 - \Psi_{G_{h-1}}\left( \ldots \Psi_{G_{2}} \left(\Psi_{G_{1}}(s)\right)\right)=
		\displaystyle{\frac{\prod_{i=1}^{h-1}\alpha_i\left(1-s\right)}{1 + \left(1-\prod_{i=1}^{h-1}\alpha_i\right)\mu\left(1-s\right) }}.$$
		
		Replace this result into $\Psi_{G_{h}}\left( \Psi_{G_{h-1}}\left( \ldots \Psi_{G_{2}} \left(\Psi_{G_{1}}(s)\right)\right)\right)$ to conclude the proof:
		\begin{align*}
		&\Psi_{G_{h}}\left( \Psi_{G_{h-1}}\left( \ldots \Psi_{G_{2}} \left(\Psi_{G_{1}}(s)\right)\right)\right)\\
		&= \frac{1 + \left( 1 - \alpha_{h}/{\mu_{\varepsilon}}_{h} \right){\mu_{\varepsilon}}_{h}\left(1 - \Psi_{G_{h-1}}\left( \ldots \Psi_{G_{2}} \left(\Psi_{G_{1}}(s)\right)\right)\right)}{1+{\mu_{\varepsilon}}_{h}\left(1 - \Psi_{G_{h-1}}\left( \ldots \Psi_{G_{2}} \left(\Psi_{G_{1}}(s)\right)\right)\right)}\\
		&= \frac{1 + \left( 1 - \alpha_{h}/{\mu_{\varepsilon}}_{h} \right){\mu_{\varepsilon}}_{h}\left(\frac{\prod_{i=1}^{h-1}\alpha_i\left(1-s\right)}{1 + \left(1-\prod_{i=1}^{h-1}\alpha_i\right)\mu\left(1-s\right) }\right)}{1+{\mu_{\varepsilon}}_{h}\left(\frac{\prod_{i=1}^{h-1}\alpha_i\left(1-s\right)}{1 + \left(1-\prod_{i=1}^{h-1}\alpha_i\right)\mu\left(1-s\right) }\right)}\\
		&=\frac{1 + \left( 1 - \alpha_{h}/{\mu_{\varepsilon}}_{h} \right){\mu_{\varepsilon}}_{h}\left(\frac{\prod_{i=1}^{h-1}\alpha_i\left(1-s\right)}{1 + \left(1-\prod_{i=1}^{h-1}\alpha_i\right)\mu\left(1-s\right) }\right)}{\left[1 + \left(1-\prod_{i=1}^{h}{\alpha_i}\right)\mu(1-s)\right]\big/\left[1 + \left(1-\prod_{i=1}^{h-1}\alpha_i\right)\mu\left(1-s\right)\right]}\\[10pt]
		&=\frac{1 + \left(1-\prod_{i=1}^{h-1}{\alpha_i}\right)\mu(1-s) + \left[ 1 - \frac{\alpha_{h}}{(1-\alpha_{h})\mu} \right](1-\alpha_{h})\mu\prod_{i=1}^{h-1}\alpha_i\left(1-s\right)}{1 + \left(1-\prod_{i=1}^{h}\alpha_i\right)\mu\left(1-s\right)}\\
		&=\frac{1 + \left[1 - \frac{\prod_{i=1}^{h}\alpha_i}{\left(1-\prod_{i=1}^{h}\alpha_i\right)\mu}\right]\left(1-\prod_{i=1}^{h}\alpha_i\right)\mu\left(1-s\right)}{1 + \left(1-\prod_{i=1}^{h}\alpha_i\right)\mu\left(1-s\right)}.
		\end{align*}
	\end{itemize}
\end{proof}

\begin{lemma}\label{prop.sum}
	Let $Z_1 \sim \text{ZMG}\left(1 - \alpha^h, \mu\right)$ and $Z_2 \sim \text{Geo}\left((1 - \alpha^h)\mu\right)$, with $\mu > 0$, $\alpha \in (0,1)$, and $h \geq 1$. Also, let $Z_1 \perp Z_2$. Then,
	$$Z_1 + Z_2 \iguald Z \sim \text{Geo}(\mu).$$
\end{lemma}
\begin{proof}
\begin{align*}
    \Psi_{Z_1+Z_2}(s) &= \Psi_{Z_{1}}(s)\Psi_{Z_{2}}(s)\\
	&=\frac{1 + (1-\alpha^h)\mu(1-s)}{1+\mu(1-s)} \cdot \frac{1}{1+(1-\alpha^h)\mu(1-s)}\\
	&=\frac{1}{1+\mu(1-s)}\\
	&=\Psi_{Z}(s).
\end{align*}
\end{proof}

\propassociatedoperatorslabel*
\begin{proof}
	We have that
	\begin{align*}
	\Psi_{\bm{\theta}_{h}\star \bm{\theta}_{h-1}\star \cdots \star \bm{\theta}_{1}\star X}(s) &= \Psi_{X}\left( \Psi_{G_{h}}\left( \Psi_{G_{h-1}}\left( \ldots \Psi_{G_{2}} \left(\Psi_{G_{1}}(s)\right)\right)\right)\right)\\
	&=\frac{1}{1+\mu\left[1-\Psi_{G_{h}}\left( \Psi_{G_{h-1}}\left( \ldots \Psi_{G_{2}} \left(\Psi_{G_{1}}(s)\right)\right)\right)\right]}.
	\end{align*}
	
 Replace the above result in Equation \eqref{eq:comppgfs} to obtain
	$$\displaystyle{\Psi_{\bm{\theta}_{h}\star \bm{\theta}_{h-1}\star \cdots \star \bm{\theta}_{1}\star X}(s) =
		\frac{1 + \left(1 - \prod_{i=1}^{k}\alpha_i \right)\mu(1-s)}{1+\mu(1-s)}.}$$
\end{proof}
\propconditionalpgflabel*
\begin{proof}
	The proof follow the same procedures used in \citep{Ristic09}.
	
	At first, by the definition of the $\star$ operator we obtain that
	\begin{align*}
	E\left( s^{X_{t+h}}\big| X_t \right) &= E\left( s^{\bm{\theta}\star X_{t+(h-1)} + \varepsilon_{t+h}}\big|X_t \right) \\
	&= E\left[E\left( s^{\bm{\theta}\star X_{t+(h-1)}}| X_{t+(h-1)} \right)\bigg| X_t\right]E\left( s^{\varepsilon_{t+h}}\big| X_t \right) \\
	&= E\left[E\left( s^{\sum_{i=1}^{X_{t+(k-1)}}{G_i}}| X_{t+(h-1)} \right)\bigg| X_t\right]\Psi_{\varepsilon}(s)\\
	&= E\left[ \prod_{i=1}^{X_{t+(h-1)}}{E\left( s^{G_i}\big| X_{t+(h-1)} \right)}\big| X_t \right]\Psi_{\varepsilon}(s) \\
	&= E\left[ \prod_{i=1}^{X_{t+(h-1)}}{E\left( s^{G} \right)}\big| X_t \right]\Psi_{\varepsilon}(s)\\
	&= E\left[ \left( \Psi_{G}(s) \right)^{X_{t+(h-1)}} \big| X_t \right]\Psi_{\varepsilon}(s).
	\end{align*}
	
	Then, after repeating $h$ times, we reach to
	$$E\left( s^{X_{t+h}}\big| X_t \right) = \prod\limits_{i=0}^{h-1}{\Psi_{\varepsilon}\left( \Psi_{G}^{(i)}(s) \right)}\cdot \left( \Psi_{G}^{(h)}(s)\right)^{X_t},$$  
	where $\Psi_{G}^{(h)}(s) = \Psi_{G}\left( \Psi_{G}^{(h-1)}(s) \right)$ and $\Psi_{G}^{(0)}(s) = s$. Applying Eq. \eqref{eq:pgfgeneral1}, we obtain that
	\begin{equation}
	E\left( s^{X_{t+h}}\big| X_t \right) = \Psi_{X}(s)\left[ \Psi_{X}\left( \Psi_{G}^{(h)}(s) \right) \right]^{-1}\left( \Psi_{G}^{(h)}(s)\right)^{X_t}.
	\label{eq:condpgf}
	\end{equation}
	
	By induction, we can prove that $\Psi_{G}^{(h)}(s) = \dfrac{1 + \left[(1-\alpha^{h-1})\mu - \alpha^{h}\right](1-s)}{1 + (1-\alpha^{h-1})\mu(1-s)}$ and substituting this in Eq. \eqref{eq:condpgf}, we finally achieve the conditional pgf as
	\begin{align*}
	E\left( s^{X_{t+h}}\big| X_t \right) =
	&\Psi_{X}(s)\left[ \Psi_{X}\left( \frac{1 + \left[(1-\alpha^{h-1})\mu - \alpha^{h}\right](1-s)}{1 + (1-\alpha^{h-1})\mu(1-s)} \right) \right]^{-1}\\
	&\times \left( \frac{1 + \left[(1-\alpha^{h-1})\mu - \alpha^{h}\right](1-s)}{1 + (1-\alpha^{h-1})\mu(1-s)} \right)^{X_t}.
	\end{align*}
\end{proof}

\clsproplabel*
\begin{proof} To prove this proposition it is enough to show that all conditions given in Theorems 3.1 and 3.2 from Tj\o stheim (1986) are satisfied.
	
	To begin note that $E\left(|X_t|^{2} \right) < \infty$ and that $E(X_t|X_{t-1}) = \alpha X_{t-1} + (1-\alpha)\mu$, as a function of $\mu$ and $\alpha$, is almost surely three times continuously differentiable in an open set $\Theta$ containing $\bm{\theta}_{0} = (\mu_{0}, \alpha_{0})$, the true value of the unknown parameter $\bm{\theta}$.
	
	\noindent \textbf{Condition 1:}
	\begin{itemize}
	    \item [] $\text{(I$_{1}$)}$:\; $E\left\{ \left| \dfrac{\partial E(X_t|X_{t-1})}{\partial \theta_{i}}(\bm{\theta}_{0}) \right|^{2} \right\} < \infty$, and
	    
	    \vspace{0.1cm}
	    
	    \item [] $\text{(II$_{1}$)}$: $E\left\{ \left| \dfrac{\partial^{2} E(X_t|X_{t-1})}{\partial \theta_{i} \partial \theta_{j}}(\bm{\theta}_{0}) \right|^{2} \right\} < \infty$
	\end{itemize}
	for $i,j=1,2$.
	
	Indeed. Note that in (I$_{1}$) we have
	\begin{equation*}
	\begin{aligned}
	E\left\{ \left| \dfrac{\partial E(X_t|X_{t-1})}{\partial \mu}(\bm{\theta}_{0}) \right|^{2} \right\} &=
	E\left\{ \left| 1 - \alpha_{0} \right|^{2} \right\} = \left( 1 - \alpha_{0} \right)^{2} < \infty \;\; \text{and} \\[10pt]
	E\left\{ \left| \dfrac{\partial E(X_t|X_{t-1})}{\partial \alpha}(\bm{\theta}_{0}) \right|^{2} \right\} &=
	E\left\{ \left| X_{t-1} - \mu_{0} \right|^{2} \right\} = Var(X_{t-1}) < \infty .
	\end{aligned}
	\end{equation*}
	
	While from (II$_{1}$) follows
	\begin{equation*}
	\begin{aligned}
	E\left\{ \left| \dfrac{\partial^{2} E(X_t|X_{t-1})}{\partial \mu^{2}}(\bm{\theta}_{0}) \right|^{2} \right\} &=
	E\left\{ \left| 0 \right|^{2} \right\} = 0 < \infty,\\[10pt]
	E\left\{ \left| \dfrac{\partial^{2} E(X_t|X_{t-1})}{\partial \mu \partial \alpha}(\bm{\theta}_{0}) \right|^{2} \right\} &=
	E\left\{ \left| -1 \right|^{2} \right\} = 1 < \infty,\\[10pt]
	E\left\{ \left| \dfrac{\partial^{2} E(X_t|X_{t-1})}{\partial \alpha^{2}}(\bm{\theta}_{0}) \right|^{2} \right\} &=
	E\left\{ \left| 0 \right|^{2} \right\} = 0 < \infty \;\; \text{and}\\[10pt]
	E\left\{ \left| \dfrac{\partial^{2} E(X_t|X_{t-1})}{\partial \alpha \partial \mu}(\bm{\theta}_{0}) \right|^{2} \right\} &=
	E\left\{ \left| -1 \right|^{2} \right\} = 1 < \infty.
	\end{aligned}
	\end{equation*}
	
	\noindent \textbf{Condition 2:} The vectors $\partial E(X_t|X_{t-1})(\bm{\theta}_{0})/\partial \theta_{i}$,  $i=1$, $2$, are linearly independent in the sense that if $a_1$ and $a_2$ are arbitrary real numbers such that
	$$E\left\{ \left|\sum\limits_{i=1}^{2}a_{i} \dfrac{\partial E(X_t|X_{t-1})}{\partial \theta_{i}}(\bm{\theta}_{0}) \right|^{2} \right\} = 0,$$
	then $a_1 = a_2 = 0$.
	
	Note that
	\begin{equation*}
	\begin{aligned}
	E\left\{ \left|a_{1} \dfrac{\partial E(X_t|X_{t-1})}{\partial \mu}(\bm{\theta}_{0}) + 
	a_{2} \dfrac{\partial E(X_t|X_{t-1})}{\partial \alpha}(\bm{\theta}_{0}) \right|^{2} \right\} &= 0 \Rightarrow\\[10pt]
	E\left\{ \left| a_1\left(1-\alpha_{0}\right) + a_2\left(X_{t-1} - \mu_{0}\right) \right|^{2} \right\} &= 0 \Rightarrow\\[10pt]
	\underbrace{a_{1}^{2}\left( 1 - \alpha_{0} \right)^{2}}_{>0} + \underbrace{a_{2}^{2}Var\left( X_{t-1} \right)}_{>0}  &= 0 \Rightarrow\\[10pt]
	a_{1}^{2}\underbrace{\left( 1 - \alpha_{0} \right)^{2}}_{>0} = 0 \;\; \text{and} \;\; a_{2}^{2}\underbrace{Var\left( X_{t-1} \right)}_{>0} &= 0 \Rightarrow\\[10pt]
	a_{1}^{2} = 0 \;\; \text{and} \;\; a_{2}^{2} &=0.
	\end{aligned}
	\end{equation*}
	Then $a_1 = a_2 = 0$.
	
	\noindent \textbf{Condition 3:} For $\bm{\theta} \in \Theta$, there exists functions $G_{t-1}^{ijk}(X_{1},\ldots, X_{t-1})$ and $H_{t}^{ijk}(X_{1},\ldots, X_{t})$ such that
	\begin{equation*}
	\begin{aligned}
	(\text{I}_{3}): T_{t-1}^{ijk}(\bm{\theta}) &= \left| \dfrac{\partial E(X_t|X_{t-1})}{\partial \theta_{i}}(\bm{\theta})\dfrac{\partial^{2} E(X_t|X_{t-1})}{\partial \theta_{j} \partial \theta_{k}}(\bm{\theta}) \right| \; \leq \; G_{t-1}^{ijk},\;\; \\
	E\left(G_{t-1}^{ijk}\right) < \infty \\[10pt]
	(\text{II}_{3}):D_{t}^{ijk}(\bm{\theta}) &= \left| \left\{X_{t} - E(X_t|X_{t-1})(\bm{\theta}) \right\}\dfrac{\partial^{3} E(X_t|X_{t-1})}{\partial \theta_{i} \partial \theta_{j} \partial \theta_{k}}(\bm{\theta}) \right| \; \leq \; H_{t}^{ijk},\;\; \\
	E\left(G_{t}^{ijk}\right) < \infty
	\end{aligned}
	\end{equation*}
	for $i,j,k=1,2$.
	
	In (I$_{3}$) realize that $T_{t-1}^{111}(\bm{\theta})=T_{t-1}^{122}(\bm{\theta})=T_{t-1}^{211}(\bm{\theta})=T_{t-1}^{222}(\bm{\theta})=0$ and that
	\begin{equation*}
	\begin{aligned}
	T_{t-1}^{112}(\bm{\theta})=T_{t-1}^{121}(\bm{\theta}) = |\alpha - 1| < 1\\[6pt]
	T_{t-1}^{212}(\bm{\theta})=T_{t-1}^{221}(\bm{\theta}) = |X_{t-1} - \mu|.
	\end{aligned}
	\end{equation*}
	If we choose $G_{t-1}^{ijk} = \left( X_{t-1} - \mu \right)^{2} + 1$, $\forall i,j,k = 1,2$ we guarantee that $T_{t-1}^{ijk}(\bm{\theta}) \leq G_{t-1}^{ijk}$; besides $E\left(G_{t-1}^{ijk}\right) = Var(X_{t-1}) + 1 < \infty$.
	
	Concerning to (II$_{3}$), $\dfrac{\partial^{3} E(X_t|X_{t-1})}{\partial \theta_{i} \partial \theta_{j} \partial \theta_{k}}(\bm{\theta}) = 0$, $\forall i,j,k =1$, $2$. 
	
	So take $H_{t}^{ijk} = 0$, $\forall i,j,k=1$, $2$ to satisfy the condition.
	
	These three conditions ensure that $\hat{\bm{\theta}}_{CLS}$ is a strongly consistent estimator to $\bm{\theta}$.
	
	Theorem 3.2 in Tj\o stheim (1986) refers to the asymptotic distribution of $\hat{\bm{\theta}}_{cls}$ and states that
	$$\sqrt{n}\left( \hat{\bm{\theta}}_{CLS} - \bm{\theta} \right) \cd \text{N}\left(\bm{0}, \Sigma \right),$$
	wherein $\Sigma = U^{-1} R U^{-1}$.
	
	The elements involved in $\Sigma$ calculation are:
	\begin{itemize}
		\item $U$:
		\begin{equation*}
		\begin{aligned}
		U &= E\left\{ \dfrac{\partial E(X_t|X_{t-1})^{T}}{\partial \bm{\theta}}(\bm{\theta})\cdot \dfrac{\partial E(X_t|X_{t-1})}{\partial \bm{\theta}}(\bm{\theta}) \right\}\\[10pt]
		&= E\left\{
		\left(
		\begin{array}{c}
		1-\alpha\\
		X_{t-1} - \mu
		\end{array}
		\right)
		\left(
		\begin{array}{cc}
		1-\alpha & X_{t-1} - \mu
		\end{array}
		\right)
		\right\}\\[10pt]
		&= \left(
		\begin{array}{cc}
		(1-\alpha)^2 & 0 \\
		0            & \mu(1+\mu)    
		\end{array}
		\right)
		\end{aligned}
		\end{equation*}
		
		\item $f_{t|t-1}(\bm{\theta})$:
		\begin{equation*}
		\begin{aligned}
		f_{t|t-1}(\bm{\theta}) &= E\left\{ \left(X_t - E(X_t|X_{t-1})\right)\left(X_t - E(X_t|X_{t-1})\right)^{T}|X_{t-1} \right\}\\[6pt]
		&= E\left\{ \left(X_t - E(X_t|X_{t-1})\right)^{2}|X_{t-1} \right\}\\[6pt]
		&=Var(X_t|X_{t-1})\\[6pt]
		&=\alpha(1-\alpha)(1+2\mu)X_{t-1} + \sigma_{\varepsilon}^{2}
		\end{aligned}
		\end{equation*}
		
		\item $R$:
		\begin{equation*}
		\begin{aligned}
		R &= E\left\{ \dfrac{\partial E(X_t|X_{t-1})^{T}}{\partial \bm{\theta}}(\bm{\theta}) f_{t|t-1}(\bm{\theta}) \dfrac{\partial E(X_t|X_{t-1})}{\partial \bm{\theta}}(\bm{\theta}) \right\}\\[10pt]
		&= E\left\{ f_{t|t-1}(\bm{\theta})\dfrac{\partial E(X_t|X_{t-1})^{T}}{\partial \bm{\theta}}(\bm{\theta})\dfrac{\partial E(X_t|X_{t-1})}{\partial \bm{\theta}}(\bm{\theta}) \right\}\\[10pt]
		&= E\left\{ f_{t|t-1}(\bm{\theta})\left(
		\begin{array}{cc}
		(1-\alpha)^2              & (X_{t-1} - \mu)(1-\alpha) \\
		(X_{t-1} - \mu)(1-\alpha) & (X_{t-1} - \mu)^{2}    
		\end{array}
		\right)\right\}
		\end{aligned}
		\end{equation*}
		\begin{itemize}
			\item [] r$_{11}$:
			\begin{equation*}
			\begin{aligned}
			r_{11} &= E\left\{ f_{t|t-1}(\bm{\theta})(1-\alpha)^2  \right\}\\
			&= E\left\{ \left[\alpha(1-\alpha)(1+2\mu)X_{t-1} + \sigma_{\varepsilon}^{2}\right](1-\alpha)^2 \right\}\\
			&= (1-\alpha)^{2}\left[\alpha(1-\alpha)(1+2\mu)\right]E\left(X_{t-1}\right) + (1-\alpha)^{2}\sigma_{\varepsilon}^{2}\\
			&= (1-\alpha)^{3}\alpha\mu + 2(1-\alpha)^{3}\alpha\mu^{2} + (1-\alpha)^{3}\mu\left[1 + (1-\alpha)\mu\right]\\
			&= (1-\alpha)^{3}\mu\left[\alpha(1+\mu) + (1+\mu)\right]\\
			&= \mu(1+\mu)(1+\alpha)(1-\alpha)^{3}
			\end{aligned}
			\end{equation*}
			\item [] r$_{12}$=r$_{21}$:
			\begin{equation*}
			\begin{aligned}
			r_{12} &= E\left\{ f_{t|t-1}(\bm{\theta})(X_{t-1} - \mu)(1-\alpha)  \right\}\\
			&= E\left\{ \left[\alpha(1-\alpha)(1+2\mu)X_{t-1} + \sigma_{\varepsilon}^{2}\right](X_{t-1} - \mu)(1-\alpha)  \right\}\\
			&= \alpha(1-\alpha)^{2}(1+2\mu)E\left\{ X_{t-1}(X_{t-1} - \mu)\right\}+
			\sigma_{\varepsilon}^{2}(1-\alpha)E\left\{ (X_{t-1} - \mu) \right\}\\
			&= \alpha(1-\alpha)^{2}(1+2\mu)\left[ E \left(X_{t-1}^{2}\right) - E\left(X_{t-1}\right)^{2} \right]\\
			&= \mu(1+\mu)(1+2\mu)(1-\alpha)^{2}\alpha
			\end{aligned}
			\end{equation*}
			\item [] r$_{22}$:
			\begin{equation*}
			\begin{aligned}
			r_{22} &= E\left\{ f_{t|t-1}(\bm{\theta})(X_{t-1} - \mu)^{2}   \right\}\\
			&= E\left\{ \left[\alpha(1-\alpha)(1+2\mu)X_{t-1} + \sigma_{\varepsilon}^{2}\right](X_{t-1} - \mu)^{2}   \right\}\\
			&= \alpha(1-\alpha)(1+2\mu)E\left\{ X_{t-1}(X_{t-1} - \mu)^{2} \right\} + \sigma_{\varepsilon}^{2}E\left\{(X_{t-1} - \mu)^{2}\right\}\\ 
			&= \alpha(1-\alpha)(1+2\mu)\left[ E\left(X_{t-1}^{3}\right) - 2\mu E\left(X_{t-1}^{2}\right) + \mu E\left( X_{t-1} \right)  \right] + \sigma_{\varepsilon}^{2}Var(X_{t-1})\\
			&= \alpha(1-\alpha)(1+2\mu)\left[ \mu + 4\mu^{2} + 3\mu^{3} \right] + (1-\alpha)\mu^{2}(1+\mu)\left[1+(1-\alpha)\mu\right]\\
			&= \alpha(1-\alpha)(1+2\mu)\mu(1+\mu)(1+\mu+2\mu) + (1-\alpha)\mu^{2}(1+\mu)\left[1+(1-\alpha)\mu\right]
			\end{aligned}
			\end{equation*}
		\end{itemize}
		\item $\Sigma$:
		\begin{equation*}
		\begin{aligned}
		\Sigma &= U^{-1}R U^{-1}\\[10pt]
		&= \left(
		\begin{array}{cc}
		\dfrac{r_{11}}{(1-\alpha)^{4}}           & \dfrac{r_{12}}{\mu(1+\mu)(1-\alpha)^{2}} \\
		\dfrac{r_{21}}{\mu(1+\mu)(1-\alpha)^{2}} & \dfrac{r_{22}}{\mu^{2}(1+\mu)^{2}}   
		\end{array}
		\right)\\[10pt]
		&= \left(
		\begin{array}{cc}
		\dfrac{\mu(1+\mu)(1+\alpha)}{1-\alpha} & (1+2\mu)\alpha \\
		(1+2\mu)\alpha                         & \dfrac{(1+\mu + 2\mu)\sigma_{G}^{2} + \sigma_{\varepsilon}^{2}}{\mu(1+\mu)}   
		\end{array}
		\right)\\
		\end{aligned}
		\end{equation*}
	\end{itemize}
\end{proof}

\end{document}